\documentclass[graybox]{svmult}

\usepackage{graphicx}
\usepackage{epstopdf}
\usepackage{amsmath}
\usepackage{color}

\usepackage[latin1]{inputenc}
\usepackage{amsfonts}
\usepackage{tikz} 
\usetikzlibrary{matrix}
\usepackage[english]{babel}
\usepackage[latin1]{inputenc}
\usepackage[T1]{fontenc}
\usepackage{ae}

\usepackage{url}

\allowdisplaybreaks[4]
\usepackage{color}
\usepackage{tikz}
\usetikzlibrary{matrix}
\usepackage{cite}           
\usepackage{mathptmx}       % selects Times Roman as basic font
\usepackage{helvet}         % selects Helvetica as sans-serif font
\usepackage{courier}        % selects Courier as typewriter font
\usepackage{type1cm}        % activate if the above 3 fonts are
\usepackage{makeidx}         % allows index generation
\usepackage{graphicx}        % standard LaTeX graphics tool
\usepackage{multicol}        % used for the two-column index
\usepackage[bottom]{footmisc}% places footnotes at page bottom

\usepackage{amssymb,amsfonts,amscd,epsf}

\newtheorem{thm}{Theorem}
\newtheorem{lem}[thm]{Lemma}

\newcommand{\eq}[2]{\begin{equation}\label{#1}#2 \end{equation}}

\newcommand{\be}{\begin{equation}}
\newcommand{\ee}{\end{equation}}
\newcommand{\bea}{\begin{eqnarray}}
\newcommand{\eea}{\end{eqnarray}}
\newcommand{\beas}{\begin{eqnarray*}}
\newcommand{\eeas}{\end{eqnarray*}}

\def\One{\mathbb{I}}

\def\jewelb{{\;\raisebox{-20mm}{\includegraphics[height=66mm]{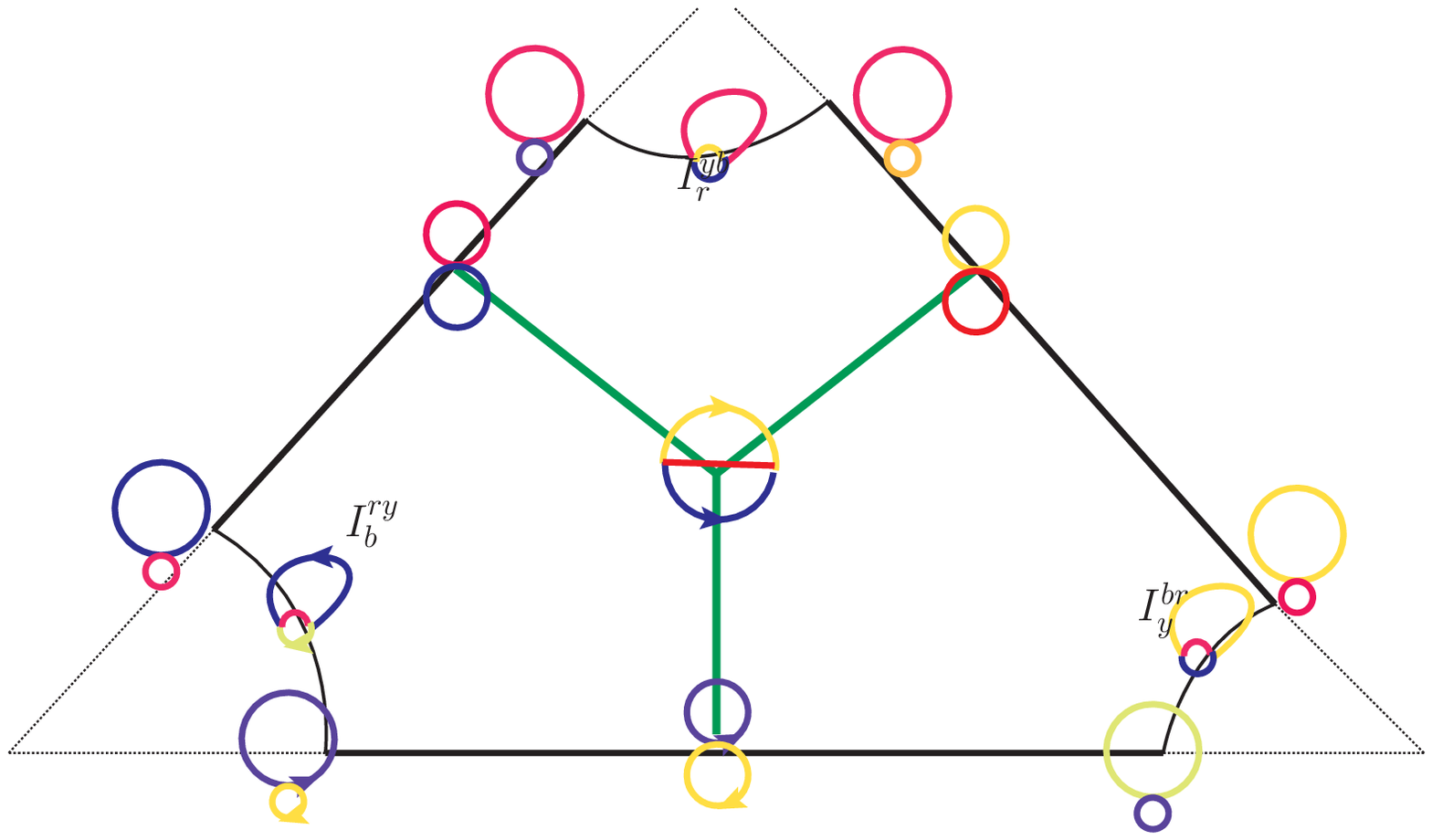}}\;}}

\def\jewelpart{{\;\raisebox{-20mm}{\includegraphics[height=90mm]{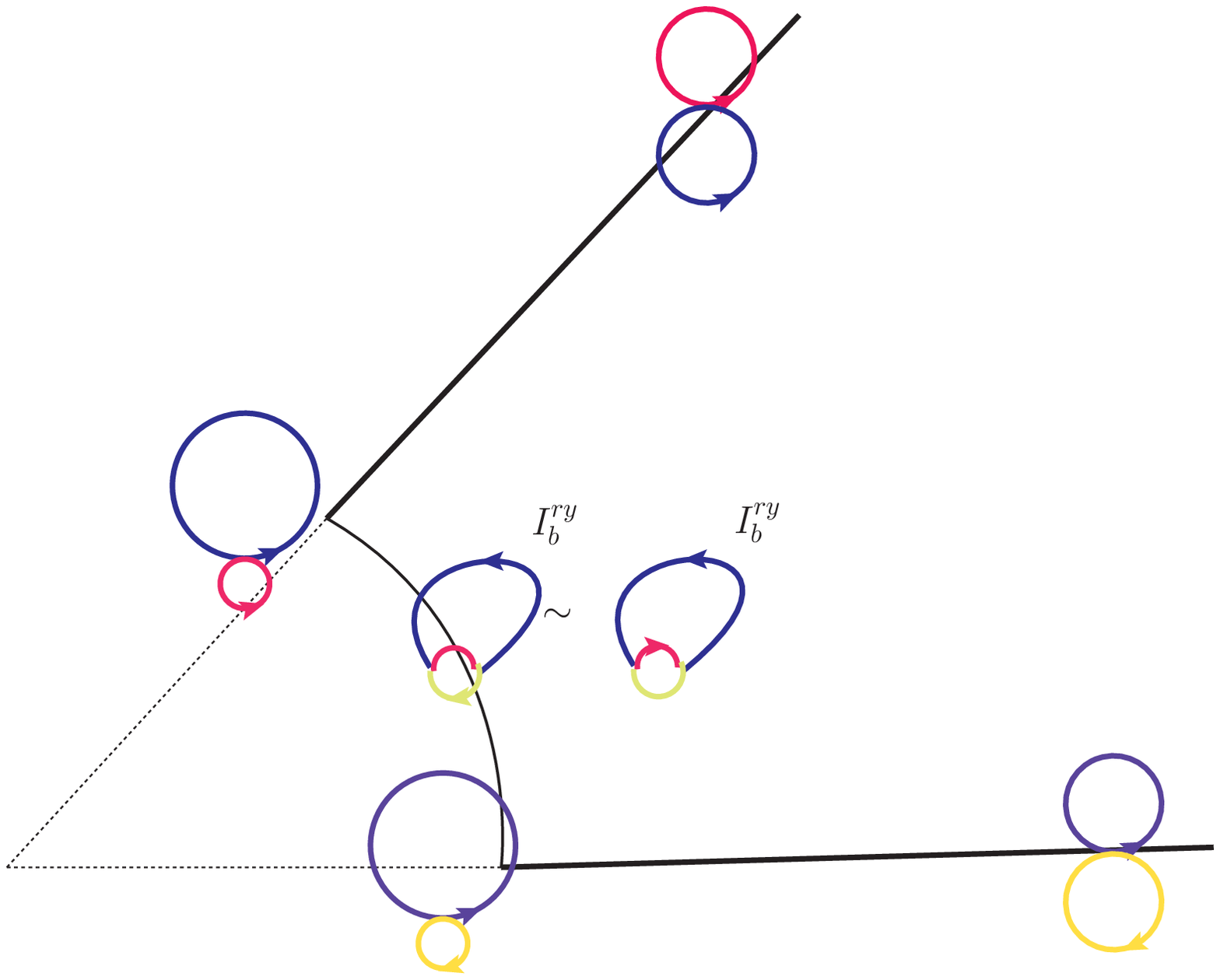}}\;}}
\def\btwo{{\;\raisebox{-10mm}{\includegraphics[height=20mm]{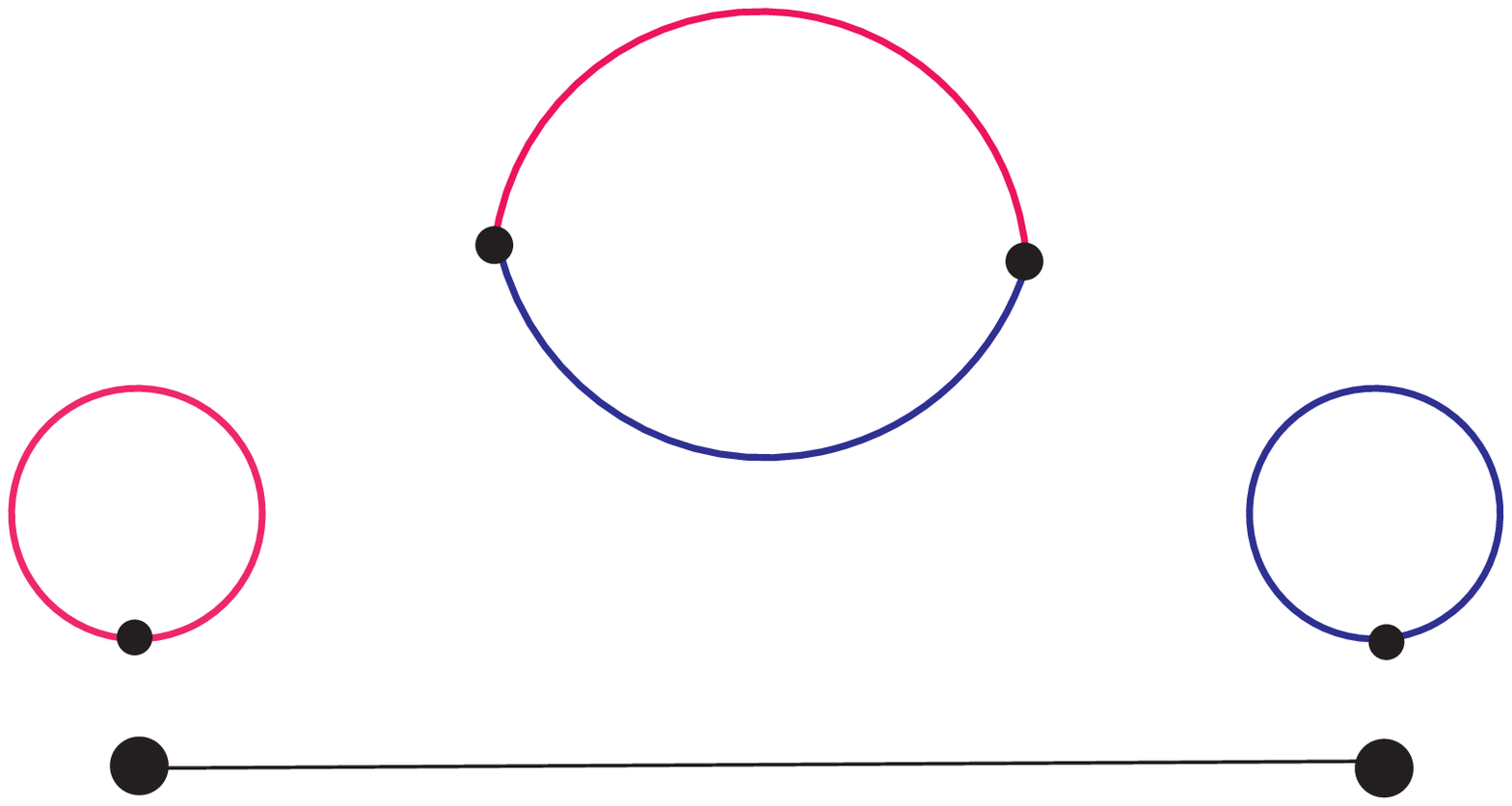}}\;}}
\def\bthree{{\;\raisebox{-10mm}{\includegraphics[height=20mm]{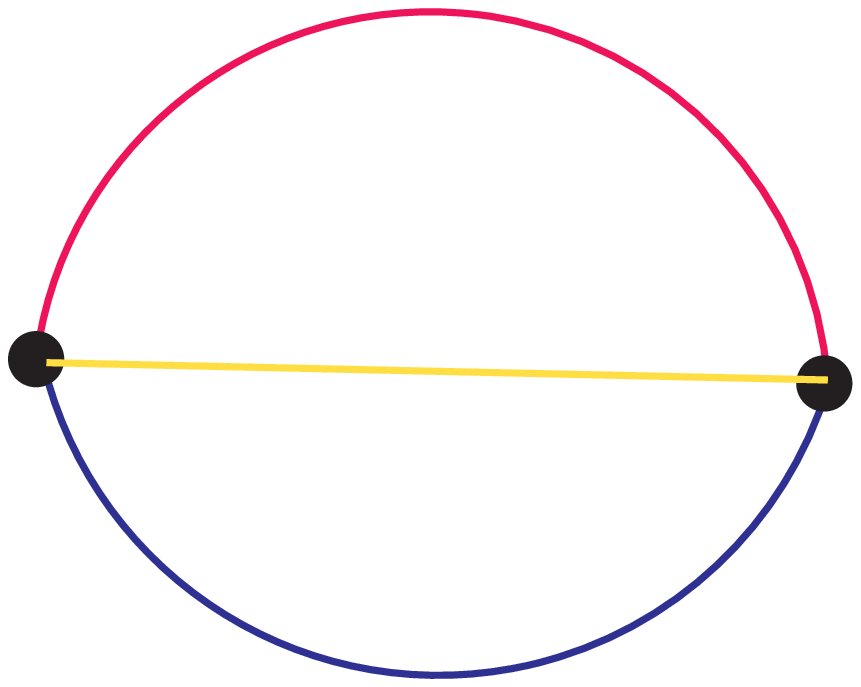}}\;}}
\def\bthreebyr{{\;\raisebox{-15mm}{\includegraphics[height=30mm]{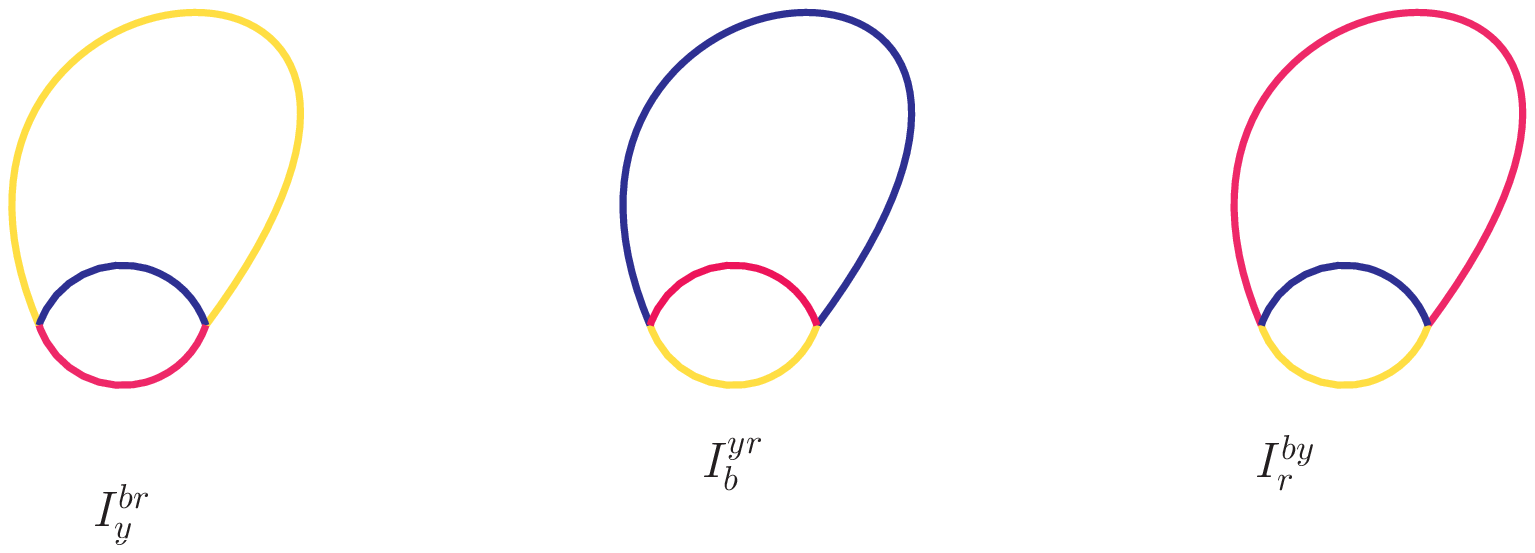}}\;}}
\def\jewelt{{\;\raisebox{-20mm}{\includegraphics[height=80mm]{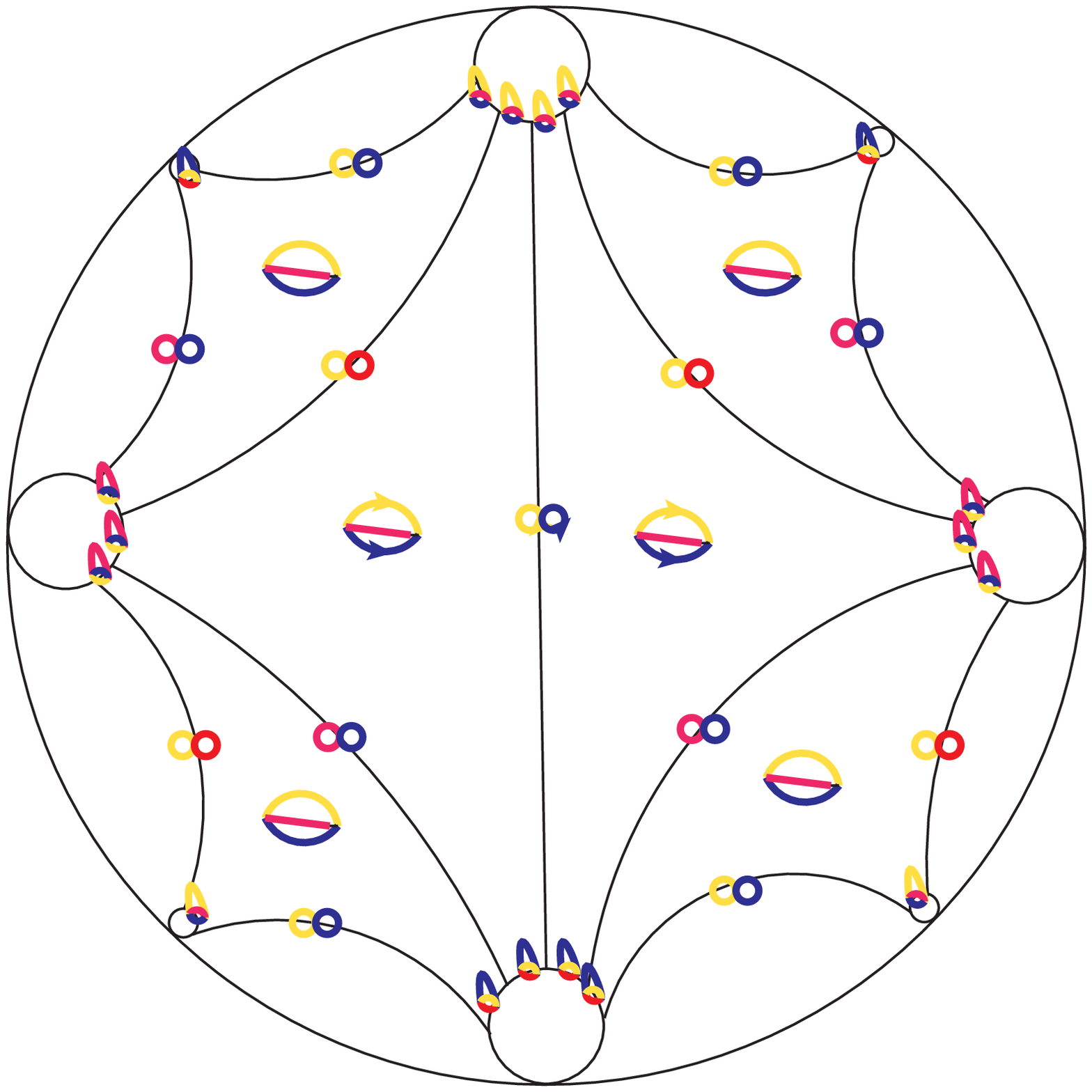}}\;}}
\def\by{{\;\raisebox{-4mm}{\includegraphics[height=10mm]{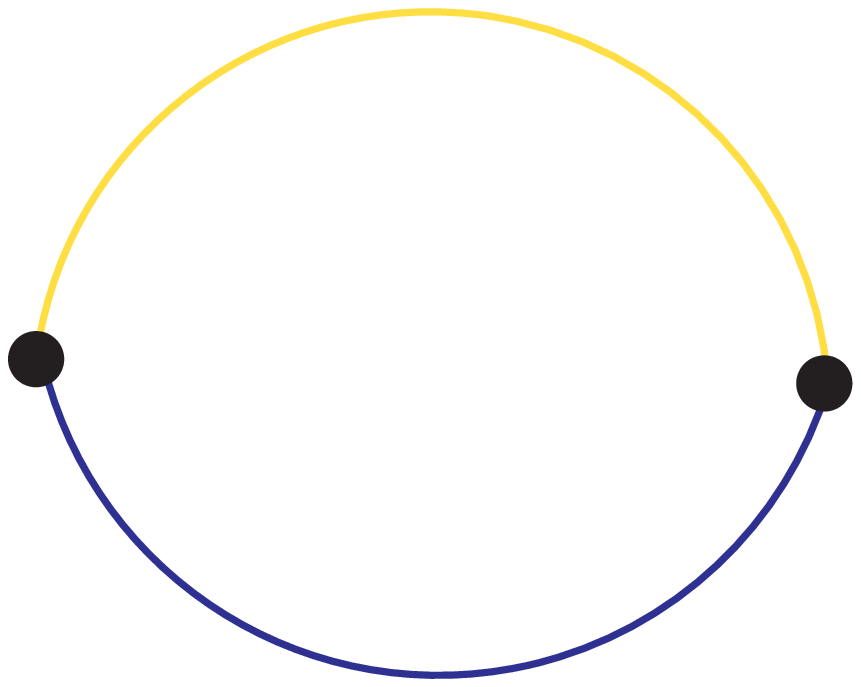}}\;}}
\def\yr{{\;\raisebox{-4mm}{\includegraphics[height=10mm]{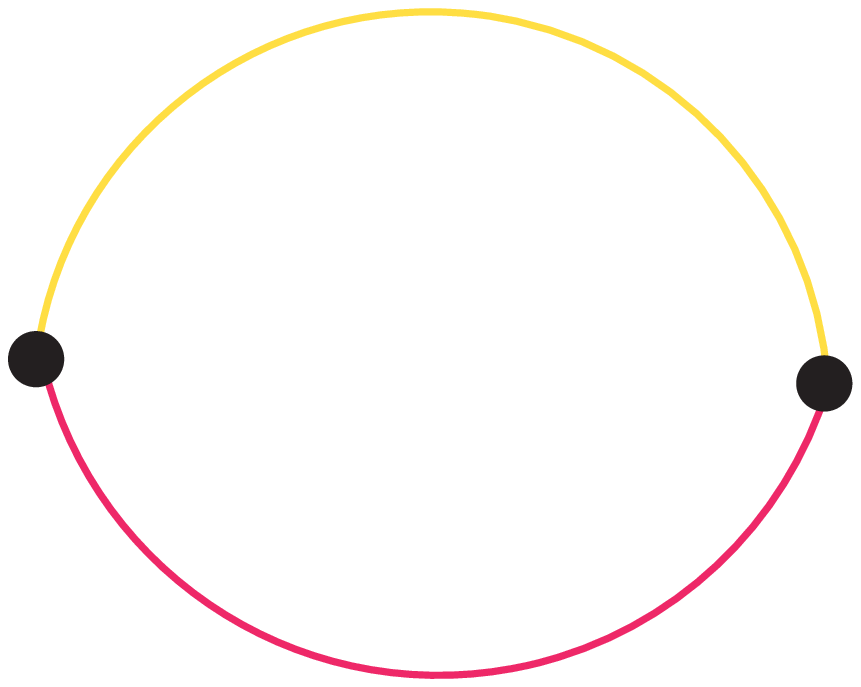}}\;}}
\def\rb{{\;\raisebox{-4mm}{\includegraphics[height=10mm]{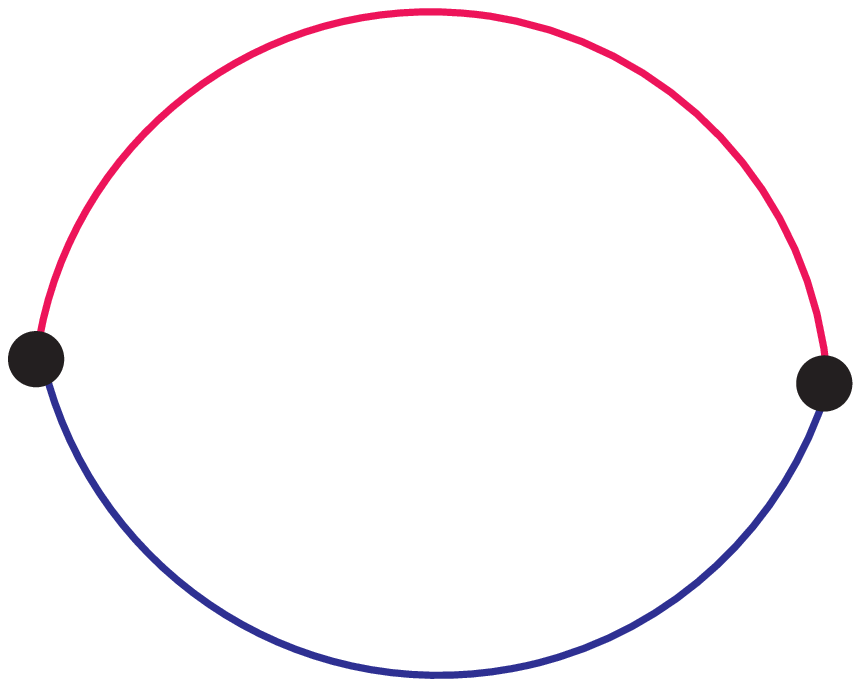}}\;}}
\def\rtp{{\;\raisebox{-4mm}{\includegraphics[height=10mm]{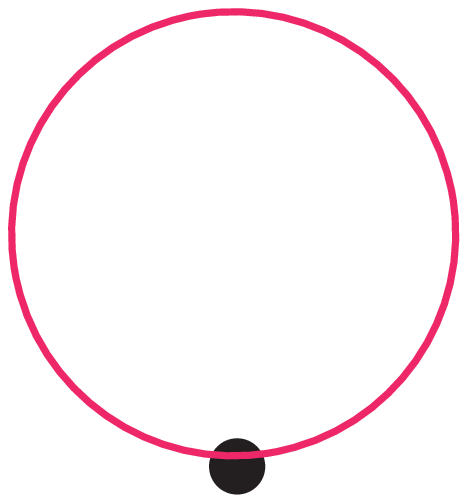}}\;}}
\def\btp{{\;\raisebox{-4mm}{\includegraphics[height=10mm]{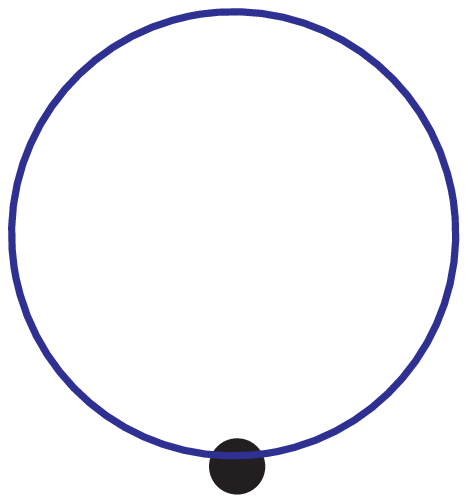}}\;}}
\def\ytp{{\;\raisebox{-4mm}{\includegraphics[height=10mm]{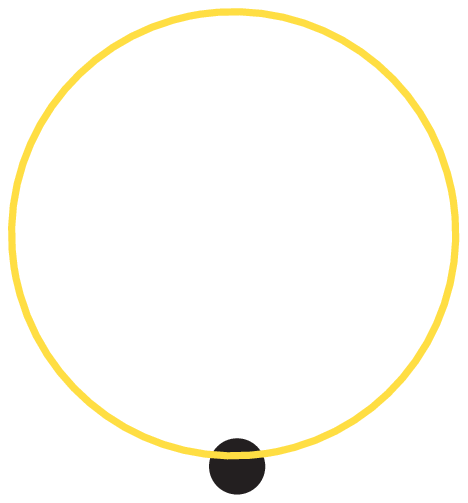}}\;}}
\def\bthreeeq{{\;\raisebox{-10mm}{\includegraphics[height=20mm]{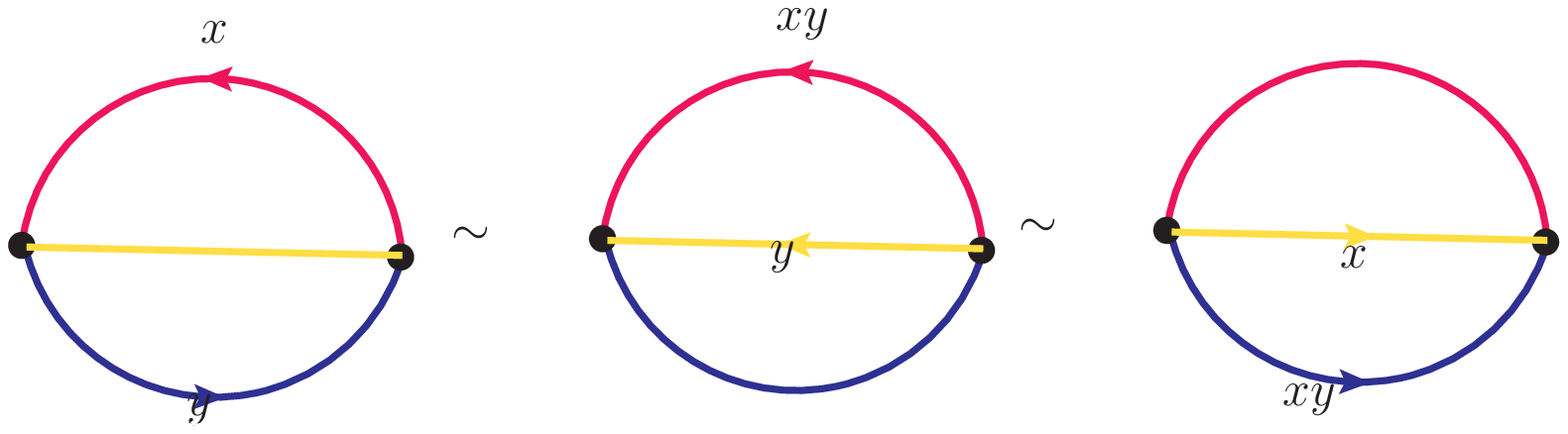}}\;}}

\def\TetraDunce{{\;\raisebox{-20mm}{\includegraphics[height=70mm]{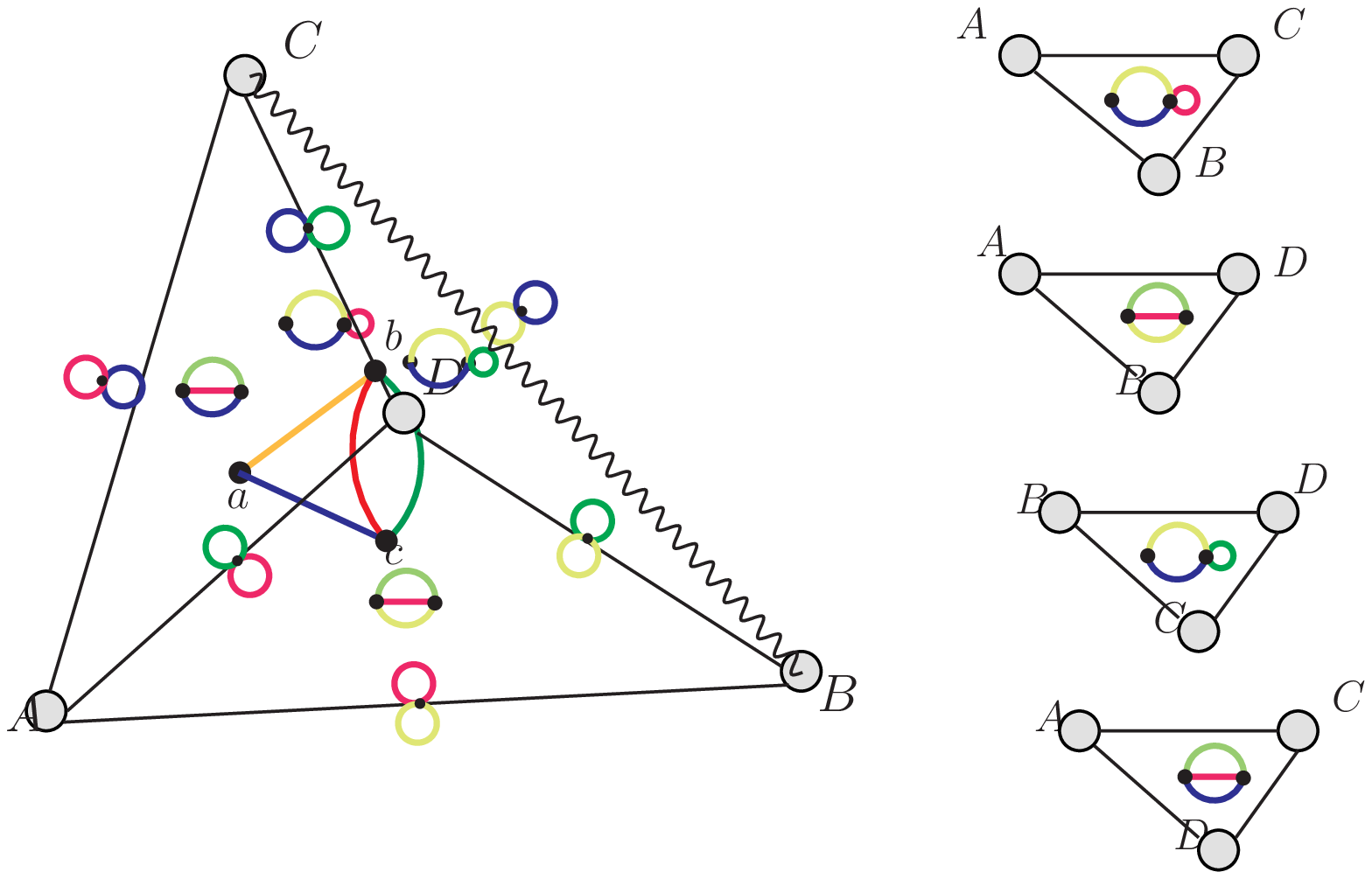}}\;}}

\def\duncetrees{{\;\raisebox{-10mm}{\includegraphics[height=24mm]{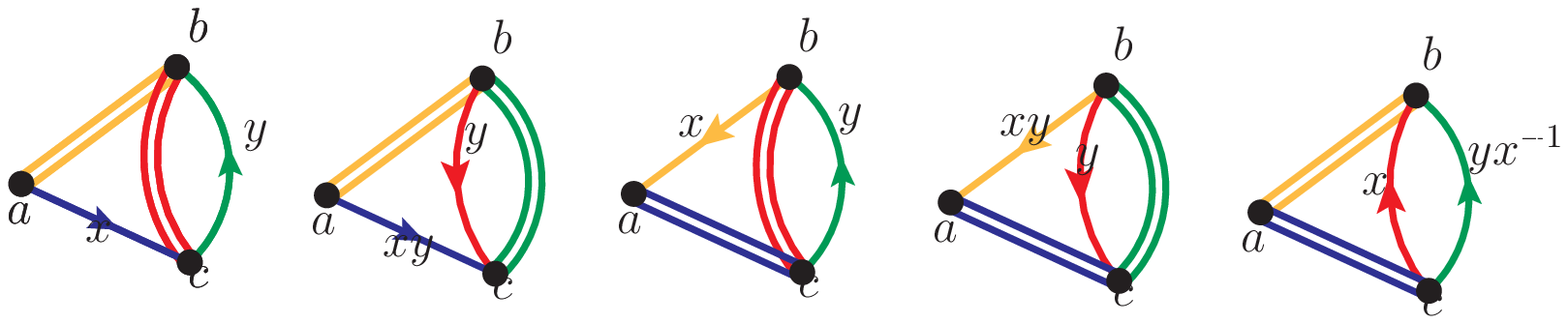}}\;}}

\def\TriangleOS{{\;\raisebox{-10mm}{\includegraphics[height=60mm]{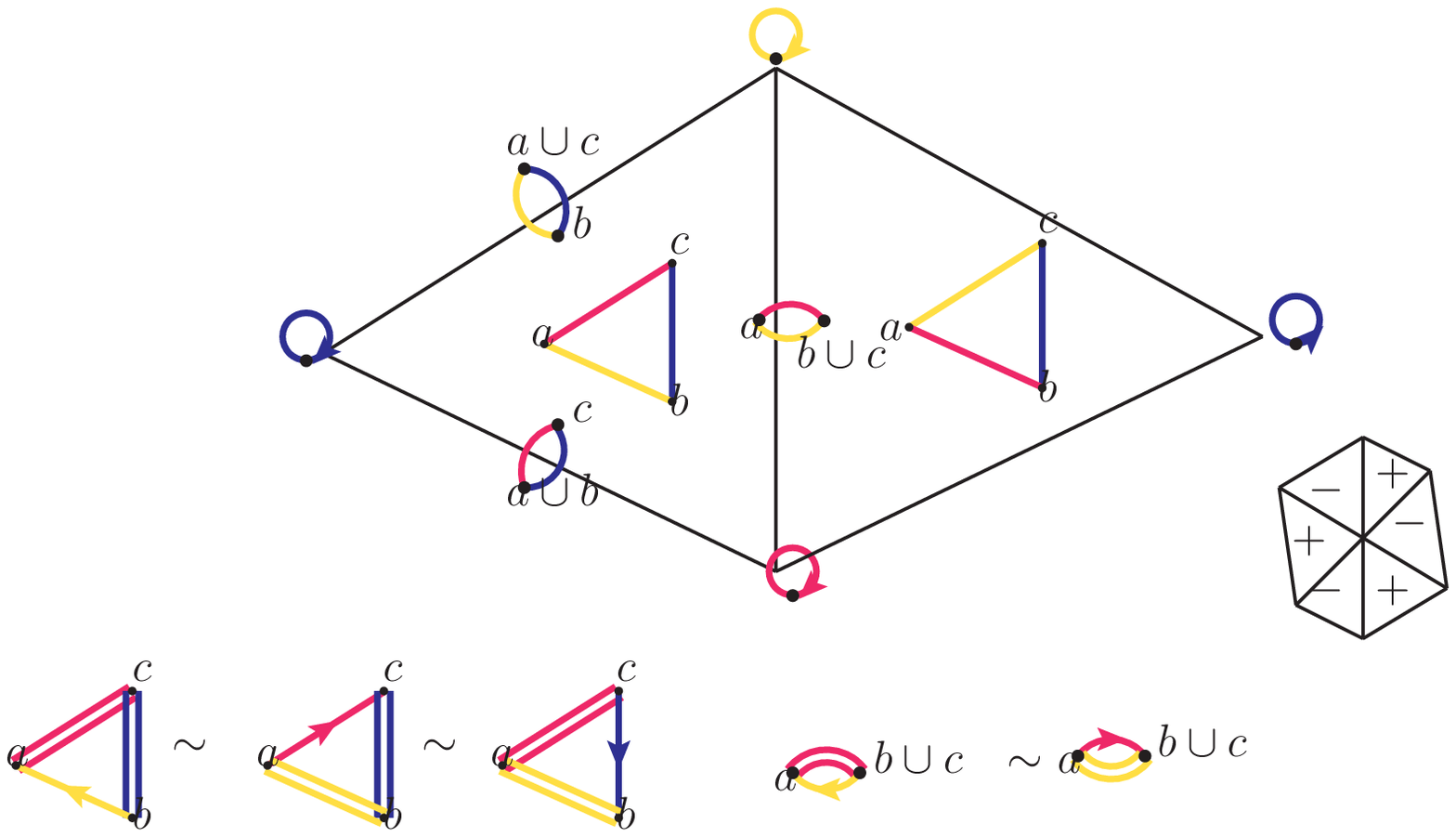}}\;}}
\def\triangle{{\;\raisebox{-1mm}{\includegraphics[height=18mm]{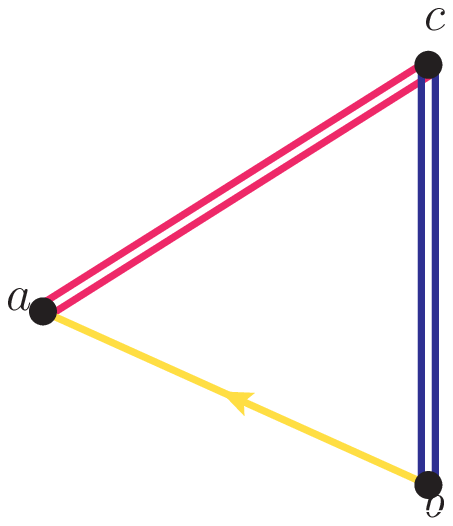}}\;}}
\def\trianglecut{{\;\raisebox{-1mm}{\includegraphics[height=18mm]{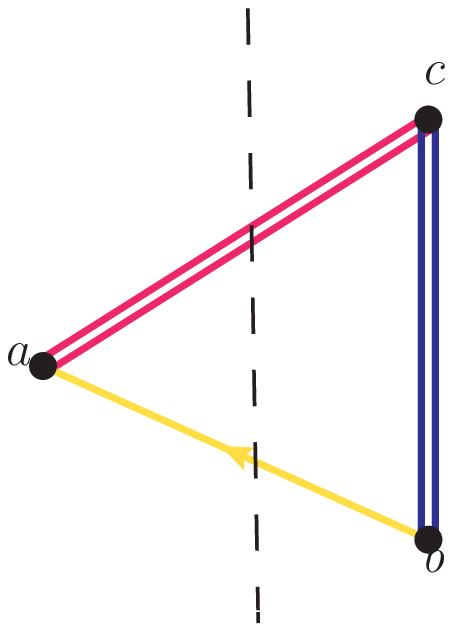}}\;}}
\def\trianglecutcut{{\;\raisebox{-1mm}{\includegraphics[height=18mm]{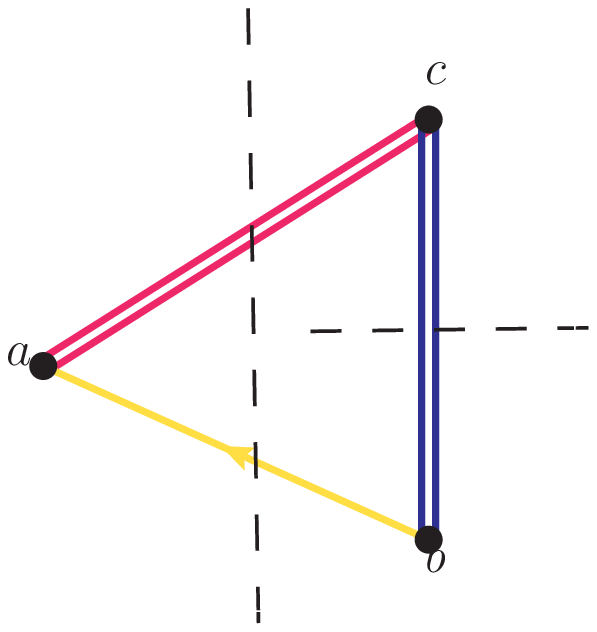}}\;}}
\def\bybubble{{\;\raisebox{-1mm}{\includegraphics[height=10mm]{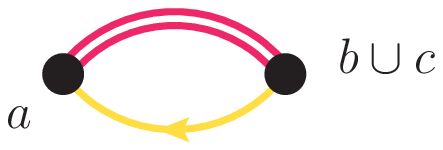}}\;}}
\def\bybubblecut{{\;\raisebox{-1mm}{\includegraphics[height=15mm]{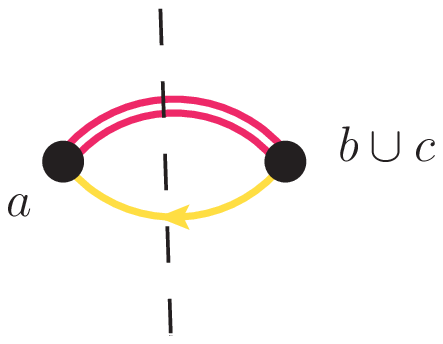}}\;}}
\def\ytad{{\;\raisebox{-1mm}{\includegraphics[height=10mm]{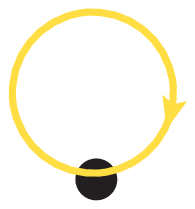}}\;}}
\def\copCCC{{\;\raisebox{-1mm}{\includegraphics[height=60mm]{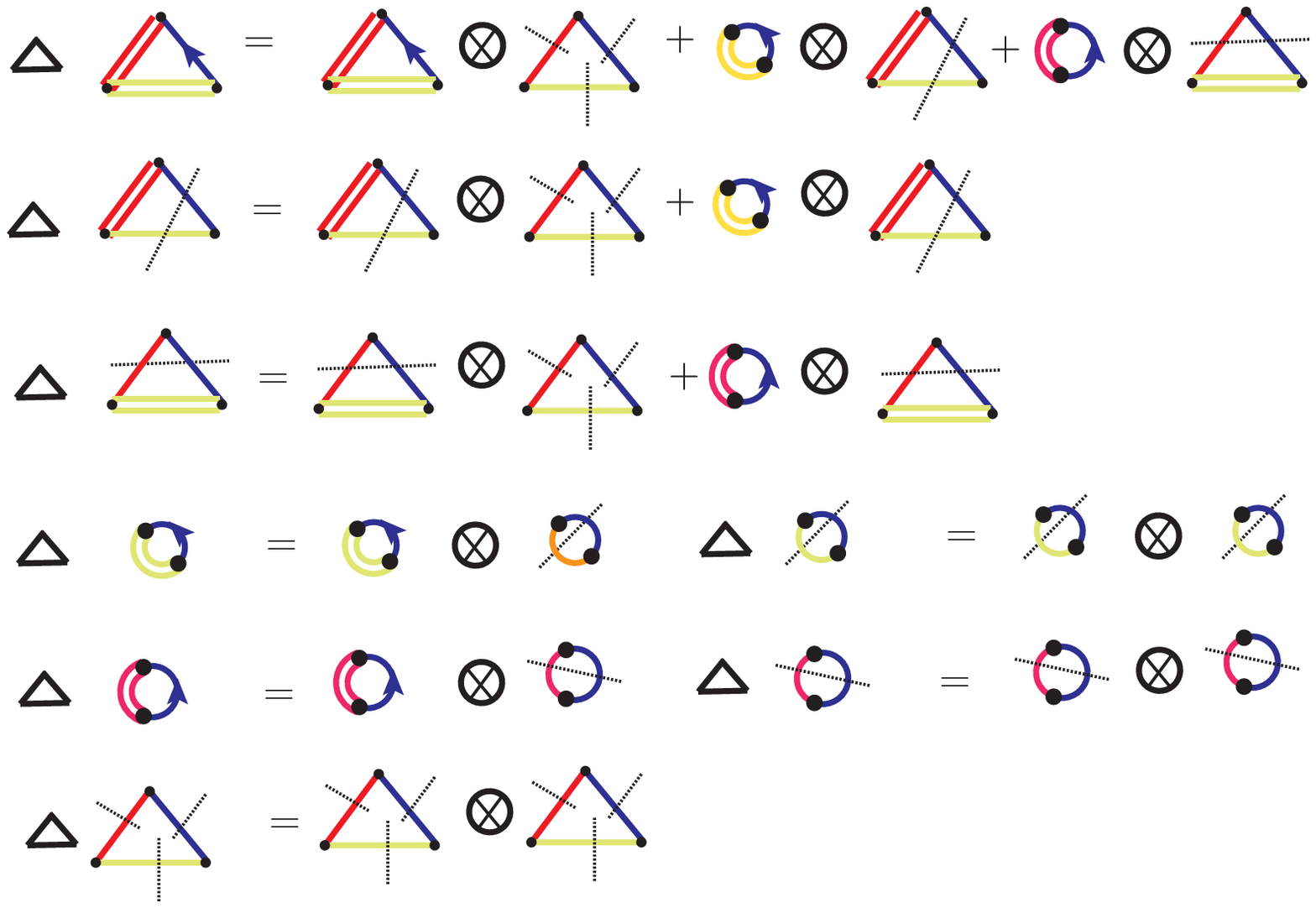}}\;}}
\def\CCCT{{\;\raisebox{-1mm}{\includegraphics[height=50mm]{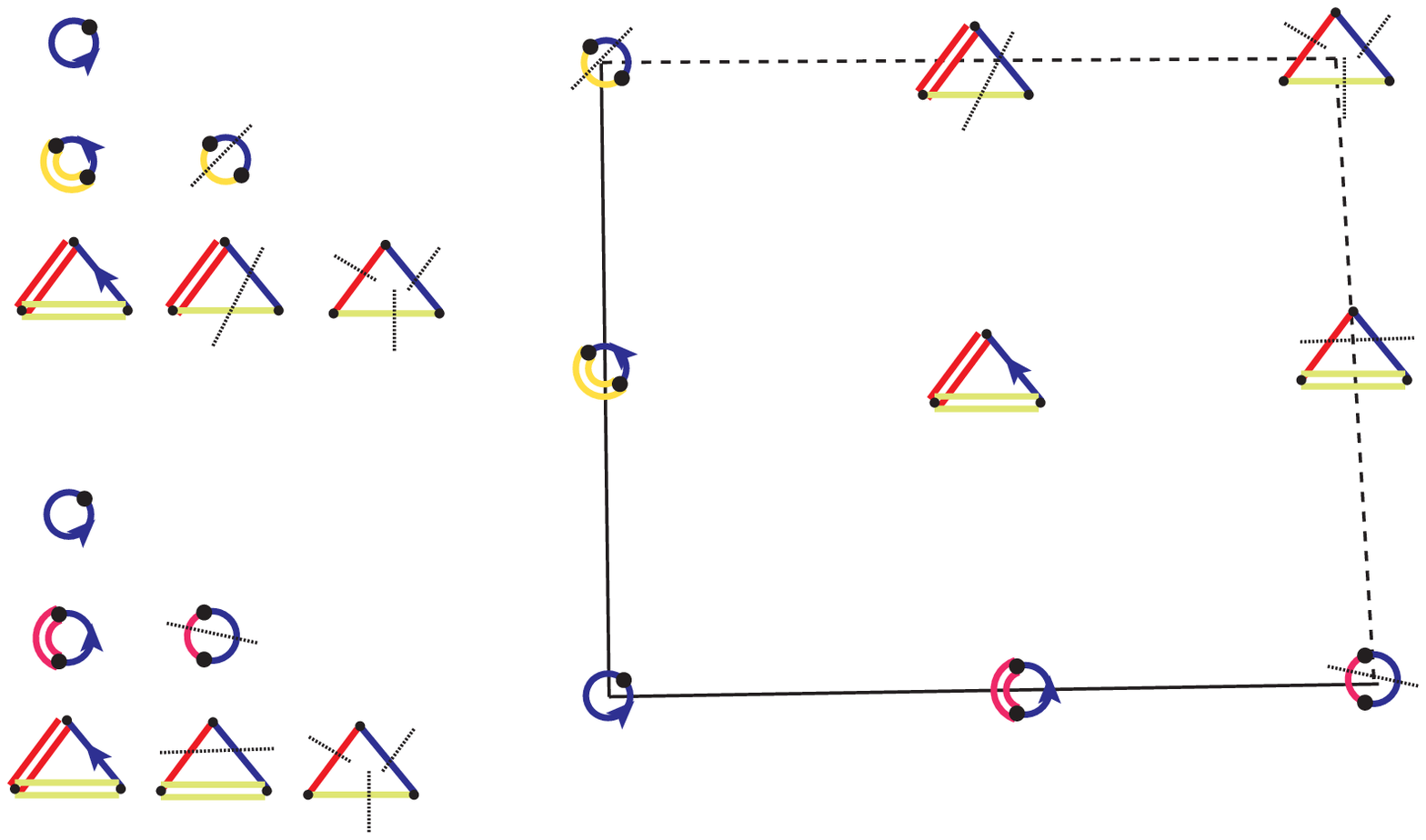}}\;}}

\def\OSCCCT{{\;\raisebox{-1mm}{\includegraphics[height=60mm]{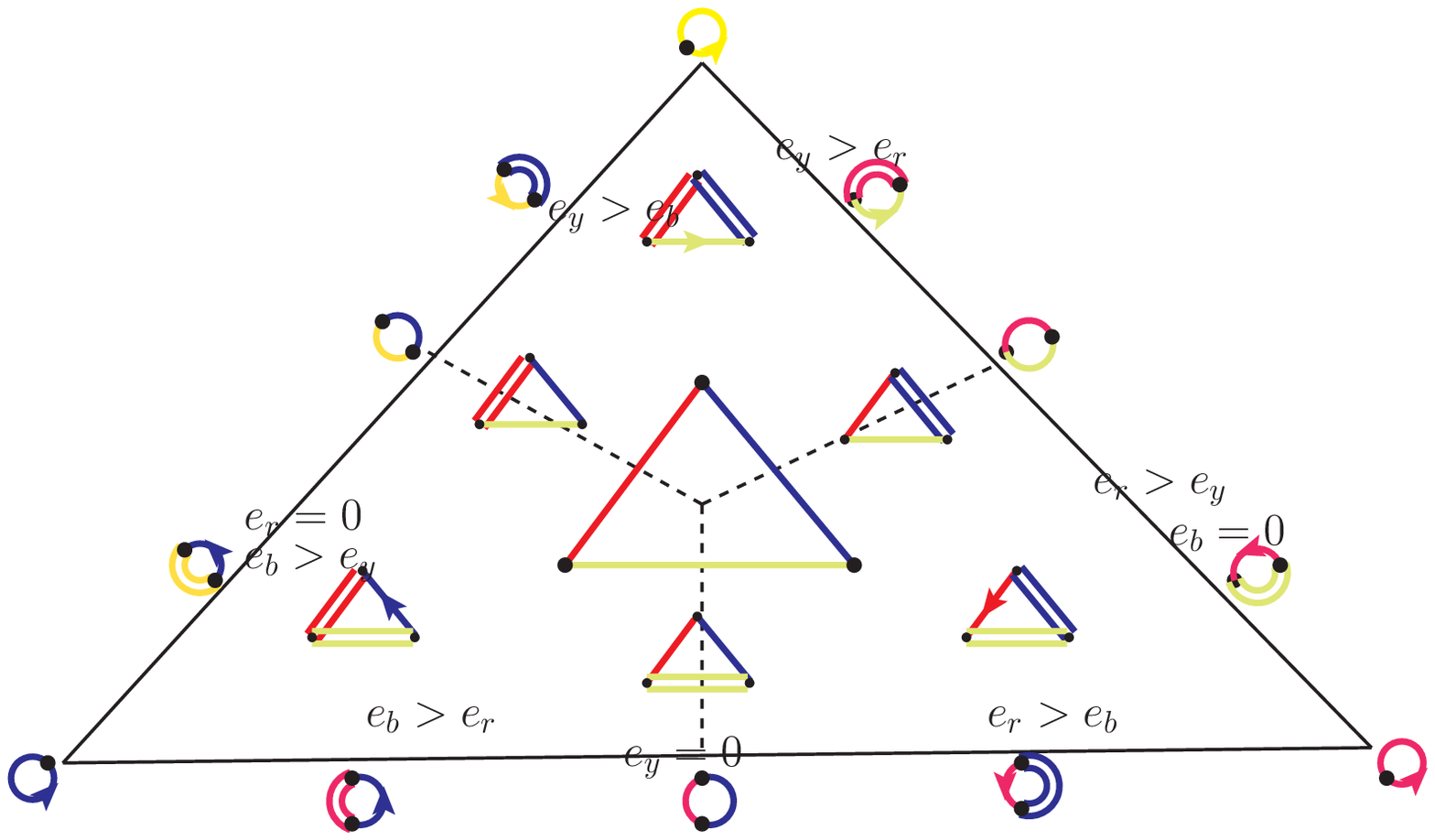}}\;}}

\makeindex             % used for the subject index
\begin{document}
\title*{Multi-valued Feynman Graphs and Scattering Theory\thanks{to appear in the Proceedings of the conference {\em Elliptic Integrals, Elliptic Functions and Modular Forms in Quantum Field Theory}, DESY-Zeuthen, October 2017, Johannes Bl\"umlein et.al., eds.}}
\author{Dirk Kreimer}
\institute{Dirk Kreimer,\\  Institute for Mathematics and Institute for Physics,
Humboldt University, Unter den Linden 6, 10099 Berlin, Germany,
\email{kreimer@physik.hu-berlin.de}}
\maketitle

%%%%%%%%%%%%%%%%%%%%%%%%%%%%%%%%%%%%%%%%%%%%%%%%%%%%%%%%%%%%%%%%%%%%%%%%%%%%%%%%%%%%%%%%%%%%%%%
\abstract{We outline ideas to connect the analytic structure of Feynman amplitudes to the structure of Karen Vogtmann's and Marc Culler's  {\em Outer Space}. We focus on the role of cubical chain complexes in this context, and also investigate the bordification problem in the example of the 3-edge banana graph.}
%%%%%%%%%%%%%%%%%%%%%%%%%%%%%%%%%%%%%%%%%%%%%%%%%%%%%%%%%%%%%%%%%%%%%%%%%%%%%%%%%%%%%%%%%%%%%%%

%
%  --> relevant terms, to appear in the index shall alwys be put into \index{....}
%

%%%%%%%%%%%%%%%%%%%%%%%%%%%%%%%%%%%%%%%%%%%%%%%%%%%%%%%%%%%%%%%%%%%%%%%%%%%%%%%%%%%%%%%%%%%%%%%
\section{Motivation and Introduction}
\label{sec:1}
\vspace{1mm}\noindent
This is a write-up of two talks given recently in Zeuthen and in Les Houches. It contains results and ideas which are partially published and which will be elaborated on in future work.
 
We want to establish a conceptual relation between scattering theory for Feynman amplitudes (see for example \cite{ErikDiss} and references there for an introduction) and the structure of suitable Outer Spaces, motivated by  \cite{KV1,KV2,KV3,KV4}. 

In particular we want to incorporate in the analysis the structure of amplitudes as multi-valued functions. We use this term for the study of functions defined by the evaluation of a Feynman graph $\Gamma\in H$ by renormalized Feynman rules, \[\Phi_R:H\to \mathbb{C}.\]
Here, $H$ is a suitable Hopf algebra of Feynman graphs. 

$\Phi_R(\Gamma),\,\forall \Gamma\in H$
depends on kinematics: $p\in Q_\Gamma$, $\Phi_R(\Gamma)=\Phi_R(\Gamma)(p)$,
where $Q_\Gamma$ is a real vectorspace \cite{BlKrCut}, 
\[ Q_\Gamma\sim\mathbb{R}^{v_\Gamma(v_\Gamma-1)/2+e_\Gamma},\] 
generated by all scalar products of external momenta, and internal masses of $\Gamma$. Here, $v_\Gamma, e_\Gamma$ are the number of vertices and edges of $\Gamma$.
We assign an external momentum to each vertex of the graph. They can be set to zero when needed.

The analysis will proceed by regarding such an amplitude $\Phi_R(\Gamma)$ as an iterated integral. What is to be iterated are not differential one-forms though, as in the example of the study of generalized polylogarithms, but elementary amplitudes $\Phi_R(\gamma)$ built from particularly simple Feynman graphs $\gamma$: one-loop graphs which form the basis primitive elements of the core Hopf algebra $H_{\mathrm{core}}$ defined below.  

$\Phi_R(\Gamma)$ is not uniquely defined as an iterated integral as there are many distinct flags \cite{BlKr1,BlKr2} which describe possible sequences of iteration to obtain $\Phi_R(\Gamma)$. 

We will remedy this below  by defining a suitable equivalence relation using equality along principal sheets. The latter equality reflects Fubini's theorem in the context of renormalized amplitudes.

Our claim is that this iteration gives $\Phi_R(\Gamma)$ -the evaluation of $\Gamma$ by renormalized Feynman rules- a structure which reflects the structure of a suitable Outer Space built on graphs.
Here, the graphs are metric marked graphs with colored edges, and without mono- or bi-valent vertices.

The full amplitude contributing at a given loop order is obtained by summing graphs which all have the same loop order and the same number of external edges. Suitably interpreted, the full amplitude is obtained as an integral over all cells of Outer Space, in a piecewise linear manner as exhibited in \cite{Marko}.

Such Outer Spaces are used in mathematics to study, amongst many things, the representation theory of the free group $F_{n_\Gamma}$.
In the course of such studies graph complexes arise which have bearing in their own right in the investigation of $\Phi_R(\Gamma)$. This includes 
\begin{itemize}
\item Outer Space itself as a cell-complex with a corresponding spine and partial order defined from shrinking edges \cite{KV1};
\item a cubical chain complex resulting from a boundary $d$ which acts on pairs $(\Gamma,F)$, $F$ a spanning forest of $\Gamma$ \cite{KV2};  
\item a bordification which blows up missing cells at infinity \cite{KV3}.
\end{itemize}
The use of metric graphs suggests itself in the study of amplitudes upon using the parametric representation: the parametric integral is then the integral over the volume of the open simplex $\sigma_\Gamma$, the cell  assigned to each graph $\Gamma$ in Outer Space, which itself is a union of such cells. 

Colored edges reflect the possibility of different masses in the propagators assigned to edges. External edges are not drawn in the coming pictures. Momentum conservation allows us to incorporate them by connecting external vertices to a distinguished vertex $v_\infty$. We come back to this elsewhere.
\section{$b_2$: the bubble}
\label{sec:2}
\vspace{1mm}\noindent
We first discuss the elementary monodromy of the simplest one-loop graph\footnote{This material is standard for physicists. It is included here for the benefit of mathematiciams who are usually not exposed to the monodromy of Feynman amplitudes.}.
So we start with the 2-edge banana $b_2$, a bubble on two edges with two different internal masses $m_b,m_r$, indicated by two different colors:
$$\btwo$$
The incoming external momenta at the two vertices of $b_2$  are $q,-q$. 

We assign to $b_2$ a one-dimensional cell, an open line segment, and glue in  its two boundary endpoints, to which the two tadpoles on the two different masses are assigned,
obtained by either shrinking the blue or red edge. The vertex at each tadpole is then 4-valent, with no external momentum flow through the graph.
 
The fundamental group \[\Pi_1(b_2)\sim\mathbb
Z\] of $b_2$ has a single generator. This matches with the monodromy of the function $\Phi_R(b_2)$ as we see in a moment.

Indeed, the Feynman integral we consider is coming from renormalized Feynman rules $\Phi_R(b_2)$, where we implement a kinematic renormalization scheme by subtraction at $\mu^2<(m_b-m_r)^2$ (so that the subtracted terms does not have an imginary part, as $\mu^2$ is even below the pseudo threshold):
\[
\Phi_R(b_2)=\int d^4k \left( \frac{1}{\underbrace{k^2-m_r^2}_{Q_1}}\frac{1}{\underbrace{(k+q)^2-m_b^2}_{Q_2}
}- \{q^2\to \mu^2\}\right).
\]
We set $s:=q^2$ and demand $s>0$, and also set $s_0:=\mu^2$. 

We write $k=(k_0,\vec{k})^T$, $t:=\vec{k}\cdot \vec{k}$. As  the 4-vector $q$ is assumed time-like (as $s>0$) we can work in a coordinate system where  $q=(q_0,0,0,0)^T$ and get 
\[
\Phi_R(b_2)=4\pi \int_{-\infty}^\infty dk_0\int_0^\infty \sqrt{t}dt \left( \frac{1}{k_0^2-t-m_r^2}\frac{1}{(k_0+q_0)^2-t-m_b^2}- \{s\to s_0\}\right).
\]
We define the K\.{a}llen function 
\[\lambda(a,b,c):=a^2+b^2+c^2-2(ab+bc+ca),\] and find by explicit integration
\beas
& & \Phi_R(b_2)(s,s_0;m_r^2,m_b^2) = \\
& &  = \left(\underbrace{ \frac{\sqrt{\lambda(s,m_r^2,m_b^2)}}{2s}\ln\frac{m_r^2+m_b^2-s-\sqrt{\lambda(s,m_r^2,m_b^2)}}{m_r^2+m_b^2-s+\sqrt{\lambda(s,m_r^2,m_b^2)}} - \frac{m_r^2-m_b^2}{2s}\ln\frac{m_r^2}{m_b^2}}_{B_2(s)}\right. \\ & & \left. -\underbrace{\{s\to s_0\}}_{B_2(s_0)}\right). 
\eeas
The principal sheet of the above logarithm is  real for $s\leq (m_r+m_b)^2$ and free of singularities at $s=0$ and $s=(m_r-m_b)^2$. It has a branch cut for $s\geq (m_r+m_b)^2$.

The threshold divisor defined by the intersection $Q_1\cap Q_2$ where the zero locii of the quadrics meet is at $s=(m_b+m_r)^2$. This is an elementary example of the application of Picard--Lefshetz theory \cite{Hwa,BlKrCut}.

Off the principal sheet, we have a pole at $s=0$ and a further branch cut for $s\leq (m_r-m_b)^2$.

It is particularly interesting to compute the variation using Cutkosky's theorem \cite{BlKrCut}
\[
\mathrm{Var}(\Phi_R(b_2))=4\pi\int_0^\infty \sqrt{z}dz \int_{-\infty}^\infty dk_0 
\delta_+(k_0^2-t-m_r^2)\delta_+((k_0-q_0)^2-t-m_b^2).
\]
Integrating $k_0$ using $k_0+q_0>0$ and $k_0=-q_0+\sqrt{t+m_b^2}$, delivers 
\[
\mathrm{Var}(\Phi_R(b_2))=\int_0^\infty \sqrt{t}dt 
\delta((q_0-\sqrt{t+m_b^2})^2-t-m_r^2)\frac{1}{\sqrt{t+m_b^2}}.
\]
With 
\[\left(q_0-\sqrt{t+m_b^2}\right)^2-t-m_r^2=s-2\sqrt{s}\sqrt{t+m_b^2}+m_b^2-m_r^2,\] 
we have from the remaining $\delta$-function,
\[
0\leq t =\frac{\lambda(s,m_r^2,m_b^2)}{4s},
\]
whenever the K\.{a}llen $\lambda(s,m_r^2,m_b^2)$ function is positive.

The integral gives
\[
\mathrm{Var}(\Phi_R(b_2))(s,m_r^2,m_b^2)=\overbrace{\left( \frac{\sqrt{\lambda(s,m_r^2,m_b^2)}}{2s}\right) }^{=:V_{rb}(s;m_r^2,m_b^2)}\times \Theta(s-(m_r+m_b)^2).
\]
We emphasize that $V_{rb}$ has a pole at $s=0$ with residue $|m_r^2-m_b^2|/2$ and 
note  $\lambda(s,m_r^2,m_b^2)=(s-(m_r+m_b)^2)(s-(m_r-m_b)^2)$.

We regain $\Phi_R(b_2)$ from $\mathrm{Var}(\Phi_R(b_2))$ by a subtracted dispersion integral:
\[
\Phi_R(b_2)=\frac{s-s_0}{\pi}\int_0^\infty \frac{\mathrm{Var}(\Phi_R(b_2)(x))}{(x-s)(x-s_0)}dx.
\]

Below we will also study the contributions of non-principal sheets and relate them to the bordification of Outer Space. 

In preparation we note that non-principal sheets give a contribution  
\[ 2j\pi \imath V_{rb}(s),\,  j\in\mathbb{Z^\times},\] where $\imath$ is the imaginary unit, $\imath^2=-1$.

We hence define a multi-valued function  
\[
\Phi_R(b_2)^\mathrm{mv}(s,m_r^2,m_b^2):=\Phi_R(b_2)(s,m_r^2,m_b^2)+2\pi\imath \mathbb{Z}V_{rb}(s).
\]
Sometimes it is convenient to write this as 
\[
\Phi_R(b_2)^\mathrm{mv}(s,m_r^2,m_b^2):=\Phi_R(b_2)(s,m_r^2,m_b^2)+2\pi\imath \mathbb{Z}\left(J_1^{rb}(s)+J_2^{rb}(s)+J_3^{rb}(s) \right),
\]
with 
\bea
J_1^{rb}(s) & := & V_{rb}(s)\Theta((m_r-m_b)^2-s),\label{jone}\\
J_2^{rb}(s) & := & V_{rb}(s)\Theta(s-(m_r-m_b)^2)\Theta((m_r+m_b)^2-s),\label{jtwo}\\
J_3^{rb}(s) & := & V_{rb}(s)\Theta(s-(m_r+m_b)^2).\label{jthree}
\eea
From the definition of the K\.{a}llen function, we conclude:\\
 $J_1^{rb}\in \mathbb{R}_+$ is positive real,\\ $J_3^{rb}(s)\in \mathbb{R}_+$ likewise  and\\ $J_2^{rb}(s)\in \imath \mathbb{R}_+$ is positive imaginary.
\section{The pole at $s=0$}
\label{sec:3}
\vspace{1mm}\noindent
In the above, we saw already a pole in $s$ appear for the evaluation along non-principal sheets for the amplitudes coming from the graph $b_2$. Actually, such poles are a general phenomenon as we want to exhibit now before we discuss the relation between the sheet structure of amplitudes and the structure of graph complexes apparent in colored variants of Outer Spaces.

We proceed using the parametric representation of amplitudes through graph polynomials as for example given in \cite{BrownKr}.

Let 
\[
\Phi_\Gamma=\phi_\Gamma+A\cdot M\psi_\Gamma,
\]
be the second Symanzik polynomial with masses 
and consider the amplitude
\[
A_\Gamma:=\int_{\mathbb{P}_\Gamma}\frac{\ln{\frac{\Phi_\Gamma}{\Phi^0_\Gamma}}}{\psi_\Gamma^2}\Omega_\Gamma.
\]
As $\psi_\Gamma$ and $\Phi^0_\Gamma$ are both strictly positive  in the domain of integration (the latter by choice of a renormalization scheme which subtracts at a kinematic point $p_0\in Q_\Gamma$ below all thresholds), we conclude
\[
\Im{A_\Gamma}:=\int_{\mathbb{P}_\Gamma}\frac{\Theta(\Phi_\Gamma)}{\psi_\Gamma^2}\Omega_\Gamma.
\]
Let $d\!A^1_\Gamma$  be the affine  measure setting $A_1=1$, and let $e\not=e_1$ be an edge $e\in E_\Gamma$, and $\mathbb{A}^1$ be the corresponding positive hypercube. 

We have $\psi_\Gamma=A_e\psi_{\gamma-e}+\psi_ {\Gamma/e}$. Then,
\[
\Im{A_\Gamma}=\int_{\mathbb{A}^1}\frac{\partial_{A_e}\Theta(\Phi_\Gamma)}{\psi_\Gamma\psi_{\Gamma-e}}d\!A^1_\Gamma+\mathrm{boundary},
\]
where we note $\partial_{A_e}\Theta(\Phi_\Gamma)=-\delta(\Phi_\Gamma)\partial_{A_e}(\Phi_\Gamma)$.
With $\Phi_\Gamma$ being a quadratic polynomial in $A_e$,
\[
\Phi_\Gamma=ZA_e^2-YA_e-X=Z\left(A_e+\frac{\tilde{Y}}{2}+\frac{1}{2}\sqrt{\tilde{Y}^2-4\tilde{X}}\right)\left(A_e+\frac{\tilde{Y}}{2}-\frac{1}{2}\sqrt{\tilde{Y}^2-4\tilde{X}}\right),
\]
we set 
\[
A_e^{\pm}:=-\frac{\tilde{Y}}{2}\pm\frac{1}{2}\sqrt{\tilde{Y}^2-4\tilde{X}}.
\]
As $\delta(f(x))=\sum_{\{x_0|f(x_0)=0\}}\frac{1}{|f^\prime(x_0)|}\delta(x-x_0)$, we get
\[
\Im A_\Gamma=\int_{\mathbb{A}^1_e}\sum_ {\pm}\frac{-\partial_{A_e}(\Phi_\Gamma)(A_e^\pm)}{|-\partial_{A_e}(\Phi_\Gamma)(A_e^\pm)|}\frac{1}{(\psi_\Gamma\psi_{\Gamma-e})(A_e^\pm)}.
\]
At $A_e=A_e^\pm$, we have $\Phi_\Gamma=0$, and therefore 
$\phi_\Gamma=A\cdot M\psi_\Gamma$, or
\[
\psi_\Gamma=\frac{\phi_\Gamma}{A\cdot M},
\]
hence
\[
\frac{1}{\psi_\Gamma\psi_{\Gamma-e}}=\frac{1}{\psi_{\Gamma-e}}\frac{A\cdot M}{\phi_\Gamma}.
\]
For a two-point function associated to a two-point graph $\Gamma$ we have  $\phi_\Gamma=s\psi_{\Gamma_\bullet}$. Here $\Gamma_\bullet$ is the graph where the two external edges of $\Gamma$ are identified. We conclude
\[
\Im A_\Gamma\sim \frac{1}{s}.
\]
For a $n$-point function, regarded as a function of a kinematical scale $s$ and angles $\vartheta_{ij}=q_i\cdot q_j/s$ \cite{BrownKr}, we find similarly\footnote{This explains the Omn\`es factor $1/\sqrt{\lambda(q_1^2,q_2^2,q_3^2)}$ in the computation of the anomalous threshold of the triangle graph.}
\[
\frac{1}{\psi_\Gamma\psi_{\Gamma-e}}=\frac{1}{s}\frac{1}{\psi_{\Gamma-e}}\frac{A\cdot M}{\phi_\Gamma(\{\vartheta_{ij}\})}.
\]

An immediate calculation gives that the boundary term remaining from the above, 
\[
\int_{\mathbb{P}_{\Gamma/e}}\frac{\Theta(\Phi_ {\Gamma/e})}{\psi_{\Gamma/e}\psi_{\Gamma-e}}\Omega_{\Gamma/e},
\]
leads to an iteration akin to linear reduction as studied by Brown and Panzer \cite{Brown1,Brown2,Erik}.

We have just proven that the two point function has a pole at $s=0$ in its imaginary part. 
This will have consequences below when we investigate in particular the analytic structure of the multi-edge banana graphs $b_3$, the story is similar for generic $b_n$.

\section{The basic set-up: Outer Space}
\label{sec:4}
\vspace{1mm}\noindent
We now first describe Outer Space. What we use is actually a variant in which there are external edges at vertices, and internal edges are colored to allow for different types of internal propagators. Here, different colors indicate generic different internal masses, but could also be used as placeholders for different spin and more\footnote{See \cite{Max} for first introductory explorations of Outer Space with colored edges in this context.}. 

\subsection{The set-up of colored Outer Space}
%The basic idea of Outer space is that points correspond to graphs with fundamental group %isomorphic to the free group $F_n$, and that $Out(F_n)$ acts by changing the isomorphism %with $F_n$. 

Outer Space can be regarded as a collection of open simplices. For a graph with $k$ edges, we assign an open simplex of dimension $k-1$. We can either demand that the sum of edge lengths (given by parametric variables $A_e$) adds to unity, or work in projective space $\mathbb{P}^{k-1}(\mathbb{R}_+)$ in such a cell. Each graph comes with a metric, and one moves around the
cell by varying the edge lengths. 

Edge lengths are allowed to
become zero but we are not allowed to shrink loops.

When an edge say between two three-valent vertices
shrinks to zero length, there are several ways to resolve the resulting 4-valent vertex to obtain a new nearby graph: assume we have a 4-valent vertex in a graph $G$ sitting in a $(k-1)$-dimensional cell. Then, this cell can be glued in as a common boundary of three other $k$-dimensional cells with corresponding graphs $G_i$, $i\in \{s,t,u\}$, which have an edge $e$ connecting two 3-valent vertices, such that $G_i/e=G$, where $G_i/e$ is the graph where edge $e$ shrinks to zero length.

For a formal definition of Outer Space we refer to \cite{KV1}. We emphasize that a crucial role is played by the fundamental group of the graph, generated by its loops. A choice of a spanning tree $T$  of a graph with $m$ independent loops  $l_i$ determines $m$ edges $e_i$  not in the spanning tree. The loops $l_i=l_i(e_i)$ are uniquely given by the edge $e_i$ and the path in $T$ connecting the two endponts of $e_i$. An orientation of $e_i$ orients the loop, and shrinking all edges of $T$ to zero length gives a rose $R_m$, a graph with one vertex and $m$ oriented petals $e_i$.  The inverse of this map gives a marking to the graph, which for us determines a choice for a basis of loops we integrate in a Feynman integral. The homotopy equivalence of such markings is reflected by the invariance of the Feynman integral under the choice how we route our momenta through the graph. 

In Outer Space graphs are metric graphs, where the metric comes from assigning an edge length to each edge, and using the parametric integrand for Feynman graph, the Feynman integral becomes an integral over the volume of the open simplex assigned to the graph, with a measure defined by the parametric representation. All vertices we assume to be of valence three or higher.

Each edge-path $l_i(e_i)$ defines a one-loop sub-integral which is multi-valued
and an ordered sequence of petals of $R_m$ defines an iterated integral of multi-valued one-loop integrals.  Using Fubini this is a well-defined integral for any ordering of the loops along the principal sheets of these loop evaluations.

Let us discuss these notions on the example of the Dunce's cap graph, to which just one of the many simplices of Outer Space is assigned.
It is a graph $\mathrm{dc}$ on four edges, accordingly, the open simplex assigned to it is a tetrahedron. The codimension-one boundaries are four triangles, to which we assign reduced graphs $\mathrm{dc}/e$ in which one of the four edges $e$ has zero length. 
$$\TetraDunce$$
The codimension-two boundaries are six edges, to which in five cases a two-petal rose $R_2$ is assigned, of the form $\mathrm{dc}/e/f$, by shrinking two of the four edges.
We can not shrink the green and red edges, as this would shrink a loop. So the edge $BC$
(indicated by a wavy line) with the rose on blue and yellow petals is actually not part of Outer Space. 

The codimension-three boundaries are the four corners $A,B,C,D$ and are not part of Outer Space either. The graph $\mathrm{dc}$ allows for five spanning trees, each of which determines a loop basis for $H^1(\mathrm{dc})$. For any of the five choices of a spanning tree, there are two edges $e,f$ say not in the spanning tree. They define two loops $l_e,l_f$ , by the edge path through the spanning tree which connects the endpoints of $e$ and $f$.

To translate this to Feynman graphs, we route all external momenta through edges of the spanning tree, and assign a loop momentum to each loop $l_e,l_f$. Any choice of order in which to carry out the loop integrations defines an iterated integral over two four-forms given by the corresponding loop integrals.   
\subsection{Example: the triangle graph} 
The above example discusses the structure of one cell together with its boundary components. We now look at the example of a triangle graph, and discuss its appearance in different cells.
$$\TriangleOS$$
Here, the boundaries of the triangular cell belong themselves to OS: the three edges of the triangular cell are a cell for the indicated 1-loop graphs on two graph-edges, the vertices correspond to colored 1-petal roses.

We have given two triangular cells in the picture. Both are associated to a triangle graph. The boundary in between is associated to the graph on a red and yellow edge as indicated. It is obtained by shrinking the blue edge. On the left- and righthand side of the boundary the triangle has permuted 
its internal red and yellow edge, with a corresponding orientation change $x\to x^{-1}$ of the single marking assigned to the graph. Gluing cells for the six possible permutations 
we obtain a hexagon, with alternating orientations as indicated.

We also give the OS equivalence relation where spanning trees are indicated by double-edges, for the triangle and for the example of the red-yellow graph on the boundary.

We omit the triangular cells corresponding to not bridge-free (not core) graphs.

Each choice of a spanning tree and choice of an ordering of its (two) edges gives rise to a Hodge matrix correponding to the evaluation of this graph as a dilogarithm \cite{BlKr1}. 
The entries are formed from the graphs apparent in the corresponding cubical cell complex
which we describe in a moment.
 
The sheet structure of the normal threshold of a two-edge graph on a boundary edge of the triangular cell is a $2\pi\imath \mathbb{Z}$ logarithmic ambiguity, the triangle  provides a further anomalous theshold which is of similar nature. 

In this way the generators of the simple fundamental group of this one-loop graph map to generators of the monodromy generated by the normal or anomalous threshold divisors of the amplitude obtained from the graph. 
\subsection{Analytic structure}
We consider the Feynman integral in momentum space and define the following quadrics.
\beas
Q_1 & := & k_0^2-t-m_1^2+i\eta,\\
Q_2 & := & (k_0+q_0)^2-t-m_2^2+i\eta,\\
Q_3 & := & (k_0+p_0)^2-t-\vec{p}^2-2\sqrt{t\vec{p}^2}z-m_1^2+i\eta,
\eeas
where $s=q_0^2$. The independent external momenta are $p$ and $q$. $q^2=q_0^2$ is time-like as before, and we compute in the rest-frame of $q$. $q$ is the momentum at vertex $a$, $p$ the momentum at vertex $c$ and $-(q+p)$ the momentum incoming at vertex $b$.

The measure $d^4k$ is transferred to $dk_0d^3\vec{k}$, and in the three-dimensional space-like part we choose spherical coordinates with $\vec{k}^2=:t$.
$z=\cos(\Theta)=\vec{p}\cdot \vec{k}/\sqrt{\vec{p}^2\vec{k}^2}$ the cosine of the angle between $\vec{k},\vec{p}$.

We are interested in the following integrals (subtractions at $s_0$ understood when necessary):\\
i) \[ \int_{\infty}^\infty dk_0 \int_0^\infty \sqrt{t}dt \int_{-1}^1dz\int_0^{2\pi}d\phi\frac{1}{Q_1Q_2} =\Phi_R(b_2),\]
ii) \[\int_{\infty}^\infty dk_0 \int_0^\infty \sqrt{t}dt
\int_{-1}^1dz\int_0^{2\pi}d\phi \delta_+(Q_1)\delta_+(Q_2)=\mathrm{Var}(\Phi_R(b_2)),\]
iii) \[\int_{\infty}^\infty dk_0 \int_0^\infty \sqrt{t}dt \int_{-1}^1dz\int_0^{2\pi}d\phi\frac{1}{Q_1Q_2Q_3}=:I_3(s,p^2,(p+q)^2;m_b^2,m_r^2,m_y^2),\]
iv) \[\mathrm{Var}^{12}(I_3):=\int_{\infty}^\infty dk_0 \int_0^\infty \sqrt{t}dt \int_{-1}^1dz\int_0^{2\pi}d\phi \delta_+(Q_1)\delta_+(Q_2)\frac{1}{Q_3}.\]
v) \[\mathrm{Var}^{123}(I_3):=\int_{\infty}^\infty dk_0 \int_0^\infty \sqrt{t}dt \int_{-1}^1dz\int_0^{2\pi}d\phi \delta_+(Q_1)\delta_+(Q_2)\delta_+(Q_3).\]

Note that for i),ii) the integrand neither depends on $z$, nor on $\phi$, so these integrals have a factor $4\pi=\mathrm{Vol}(S^2)$ in their evaluation.

Most interesting are the integrals in  iii)-v). $I_3$ is a dilogarithm, see \cite{BlKr1} for its properties.

Let us start with $\mathrm{Var}^{12}(I_3)$. The two $\delta_+$-functions constrain the $k_0$- and $t$-variables, so that the remaining integrals are over the compact domain $S^2$.

As the integrand does not depend on $\phi$, this gives a result of the form
\[ 2\pi C \underbrace{\int_{-1}^1 \frac{1}{\alpha+\beta z}dz}_{:=J(z)}=2\pi C\frac{\ln\frac{\alpha+\beta}{\alpha-\beta}}{\beta},
\]
where $C$ is intimitaly related to $\mathrm{Var}(\Phi_R(b_2))=2C$, and the factor $2$ here is $
\mathrm{Vol}(S^2)/\mathrm{Vol}(S^1)$.

Then, we get a Hodge matrix for a triangle graph $\Delta$ 
\[
\left(
\begin{matrix}
1 & 0 & 0\\
\Phi_R(b_2)(s,m_r^2,m_y^2) & V_{ry}(s,m_r^2,m_y^2) & 0\\
I_3 &    V_{ry}(s,m_r^2,m_y^2)\frac{\ln\frac{\alpha+\beta}{\alpha-\beta}}{\beta}      & \underbrace{\frac{1}{\sqrt{s,p^2,(p+q)^2}}}_{=V_{ry}(s,m_r^2,m_y^2)\times \mathrm{Var}J(z)}
\end{matrix}
\right)
\]
$\sim$
\[
\left(
\begin{matrix}
\ytad & 0 & 0\\
\bybubble & \bybubblecut & 0\\
\triangle & \trianglecut & \trianglecutcut
\end{matrix}
\right)
\]

Here, $\alpha$ and $\beta$ are given through $l_1:=\lambda(s,p_b^2,p_c^2)$ and $l_2:=\lambda(s,m_y^2,m_r^2)$ as
\[
\alpha:= (m_y^2-m_r^2-s-p_a.p_c)^2-l_1-l_2,\,\beta:=2\sqrt{l_1l_2}.
\]
The amplitude of the triangle graph in a chosen triangular cell is the lower left entry in this Hodge matrix. The leftmost column of this matrix was obtained by first shrinking the blue edge, and then the red one. This fixes the other columns which are defined by the Cutkosky cuts -the variations- of the column to the left. 

The triangle graph has a single loop and its fundamental group a single generator. 
Accordingly, we find a single generator for the monodromy in the complement of the threshold divisors: either for the normal threshold at $s_0=(m_r+m_y)^2$ or for  the anomalous threshold at $s_1$, with $l_r=p^2-m_r^2-m_b^2,l_y=(p+q)^2-m_y^2-m_b^2$, $\lambda_1=\lambda(p^2,m_r^2,m_bv^2)$, $\lambda_1=\lambda((p+q)^2,m_y^2,m_bv^2)$ it is \cite{BlKrCut} given as, 
\be 
s_1=(m_r+m_y)^2+\frac{4m_b^2(\sqrt{\lambda_2}m_r-\sqrt{\lambda_1}m_y)^2-(\sqrt{\lambda_1}l_y+\sqrt{\lambda_2}l_r)^2}{4m_b^2\sqrt{\lambda_1}\sqrt{\lambda_2}}.
\ee
 The function $J(z)$ has no pinch singularity and does not generate a new vanishing cycle. In general, a one-loop graph generates one pinch singularity through its normal threshold given by a reduced graph $b_2$, and as many anomalous thresholds as there are further edges in the graph. 

This structure iterates upon iterating one-loop graphs to multi-loop graphs. For multi-loop graphs, we discuss later only the example of the three-edge banana $b_3$, which only has a normal cut on the principal sheet, but a rather interesting structure on other sheets. The general picture will be discussed elsewhere. An algorithm to compute  anomaloues thresholds as $s_1$ is contained in \cite{BlKrCut}.

Returning to the triangle graph, there are three different spanning trees on two edges  for the triangle graph, and for each spanning tree two possibilities which edge to shrink first. This gives us six such matrices. To see the emergence of such matrices from the set-up of Outer Space turn to the cubical chain complex associated to the spine of Outer Space \cite{KV2}. 
\section{The cubical chain complex}
\label{sec:5}
\vspace{1mm}\noindent
Consider the cell (itself an open triangle)  assigned to one triangle graph. Let us assume we put the graph in the barycentric middle of the cell. At the codimension-one boundaries of the cell we glue edges, and put the corresponding graphs in their (barycentric) middle. 

These boundaries correspond to edges $e_i=0$, $i\in\{r,b,y\}$ as indicated in the figure.

At the codimension-two corners we find tadpoles. When we move the triangle towards the blue corner say, its spanning tree must be on the yellow and red edges (edges in spanning trees are indicated by a double edge in the figure), which are the ones allowed to shrink. 

When we move towards say the barycentric middle of the boundary defined by $e_y=0$, only the yellow edge is part of a spanning forest, to the left of it the red edge is spanning as well, to the right the blue one.

The dashed lines connecting the barycentre of the triangular cell with the barycentres of its codimension-one boundaries partitions the triangular cell into three regions. Each such region has four corners and four line segments connections them, and an interior, to which pairs of the triangle graph $\Delta$ and a spanning tree are assigned as indicated.
$$\OSCCCT$$
Such a decomposition of sectors exists for any cell in Outer space. If a graph $\Gamma$ has $m$ spanning trees on $n$ edges, we have $m\times n!$ paths from the barycentre of the 
$\Gamma$-cell to a rose. To this, we can assign $m$ cubes, which decompose into $n!$ simplices, and we get $m\times n!$ Hodge matrices as well, $n!$ for each cosen spanning tree, for example two for a pair of a triangle and a chosen spanning tree (containing two edges)  for it: 
$$\CCCT$$
In the figure, we have marked the edges connecting different components of a spanning forest by cuts. 
The two triangular Hodge matrices on the left in the figure correspond to the two possible choices which edge to shrink first. Both Hodge matrices contain the three graphs as entries which populate the diagonal of the cube. The other entries are from above or below the diagonal.

The entries of the cube, and therefore the entries of the Hodge matrices describing variations of the accompanying Feynman integrals, are generated from a boundary operator $d$, $d\circ d=0$, which acts on pairs $(\Gamma,F)$ of a graph and a spanning forest \cite{KV2}. 
\[
d(\Gamma,F)=\sum_{j=1}^{|E_F|}(-1)^j\left((\Gamma/e,F/e)\otimes (\Gamma,F\backslash e)\right)
\]
In our Hodge matrices, the left-most colums are distinguished, as in them only pairs $(\Gamma,F)$ appear in which $F$ is a spanning tree of $\Gamma$, while all other entries have spanning forests consisting of more than a single tree, corresponding to graphs with Cutkosky cuts.

This suggests to bring this into a form of coaction (see also \cite{Britto}), which looks as follows:
\be \copCCC\label{coaction}\ee 
Before we comment on this in any further detail, we have to collect a few more algebraic properties of Feynman graphs.
\section{Hopf algebra structure for 1PI graphs}
\label{sec:5H}
\vspace{1mm}\noindent
We start with the renormalization and core Hopf algebras.
\subsection{Core and renormalization Hopf algebras}
Consider the free commutative $\mathbb{Q}$-algebra 
\be 
H=\oplus_{i\geq 0} H^{(i)},\, H^{(0)}\sim \mathbb{Q}\One,
\ee
generated by 2-connected graphs as free generators (disjoint union is product $m$, labelling of edges and of vertices by momenta as declared).

Consider the Hopf algebras $H(m,\One,\Delta,\hat{\One},S)$ and $H(m,\One,\Delta_{\mathrm{core}},\hat{\One},S_c)$, given by
\be \One:\mathbb{Q}\to H, q\to q\One,\ee
\be \Delta:H\to H\otimes H, \Delta(\Gamma)=\Gamma\otimes\One+\One\otimes\Gamma+\sum_{\gamma\subsetneq\Gamma,\gamma=\cup_i\gamma_i,w(\gamma_i)\geq 0}
\gamma\otimes \Gamma/\gamma,\ee
\be \Delta_{\mathrm{core}}:H\to H\otimes H, \Delta_{\mathrm{core}}(\Gamma)=\Gamma\otimes\One+\One\otimes\Gamma+\sum_{\gamma\subsetneq\Gamma,\gamma=\cup_i\gamma_i}
\gamma\otimes \Gamma/\gamma,\ee
\be \hat{\One}:H\to\mathbb{Q}, q\One\to q, H_>\to 0,\ee
\be S: H\to H, S(\Gamma)=-\Gamma-\sum_{\gamma\subsetneq\Gamma,\gamma=\cup_i\gamma_i,w(\gamma_i)\geq 0}
S(\gamma) \Gamma/\gamma,\ee
\be S_{\mathrm{core}}: H\to H, S(\Gamma)=-\Gamma-\sum_{\gamma\subsetneq\Gamma,\gamma=\cup_i\gamma_i}
S_{\mathrm{core}}(\gamma) \Gamma/\gamma,\ee
where $H_>=\oplus_{i\geq 1} H^{(i)}$ is the augmentation ideal.

Both Hopf algebras will be needed in the following for renormalization in the presence of variations. 

They both have a co-radical filtration, which for the renormalization Hopf algebra delivers the renormalization group, and for the core Hopf algebra the flags of all decompositions of a graph into iterated integrals of one-loop graphs.
We often use Sweedler's notation: $\Delta\Gamma=\Gamma^\prime\otimes\Gamma^{\prime\prime}.$
\subsection{Hopf algebras and the cubical chain complex of graphs}
Let us return to graphs with spanning forests.
For a spanning tree of length $j$, there are $j!$ orderings of it edges. To such a spanning tree, we assign a $j$-dimensional cube
and to each of the $j!$ ordering of its edges a matrix as follows. We follow \cite{BlKrCut}.

Let $(\Gamma,T)$ be a pair of a graph and a spanning tree for it with a choice of ordering for its edges.
Let $\mathcal{F}_{(\Gamma,T)}$  be the set of corresponding forests obtained by removing edges from $T$ in order.

Then, to any pair $(\Gamma,F)$, with $F$ a $k$-forest ($1\leq k\leq v_\Gamma$), $F\in \mathcal{F}_{(\Gamma,T)}$ we can assign a set of $k$ disjoint graphs $\Gamma^F$. We let $\Gamma_F:=\Gamma/\Gamma^F$ be the graph obtained by shrinking all internal edges of these graphs.

For each such $F$, we call $E_{\Gamma_F}$ a cut. In particular, for $F$ the unique 2-forest assigned to $T$ (by removing the first edge from the ordered edges of $T$), we call $\epsilon_2=E_{\Gamma_F}$ the normal cut of $(\Gamma,T)$. 

Note that the ordering of edges defines an ordering of cuts
$\emptyset=\epsilon_1\subsetneq\epsilon_2\subsetneq\cdots\subsetneq\epsilon_k=E_\Gamma$.   

For a normal cut, we have $\Gamma^F=(\Gamma_1,\Gamma_2)$ and we call 
\be
s=(\sum_{v\in V_{\Gamma_1}}q_v)^2=(\sum_{v\in V_{\Gamma_2}}q_v)^2
\ee
the channel associated to $(\Gamma,T)$.

These notions are recursive in an obvious way: the difference between a $k$ and a $k+1$ forest defines a normal cut for some subgraph.

We then get a lower triangular matrix with entries from pairs $(\Gamma,F)$ by shrinking edges of the spanning tree from bottom to top in order, and removing edges from the spanning tree from left to right in reverse order. 

To set up Feynman rules for pairs $(\Gamma,F)$ we need an important lemma.

We define $|\Gamma^F|=\sum_{\gamma\in \Gamma^F}|\gamma|$. Also, we let $\mathcal{F}_k(\Gamma)$ be the set of all $k$-forests for a graph $\Gamma$.

For a disjoint union of $r$  graphs $\gamma=\cup_{i=1}^r\gamma_i$, we say a disjoint union of trees $T=\cup_i t_i$ spans $\gamma$
and write $T|\gamma$, if $t_i$ is a spanning tree for $\gamma_i$. 

We have then an obvious decomposition of all possible spanning forests using the coproduct $\Delta_{\mathrm{core}}$. A spanning forest decomposes into a spanning forest which leaves no loop intact in the cograph together with spanning trees for the subgraph \cite{core,Bley}: 
\begin{lem}\label{cccuts}
\be 
\sum_{T|\Gamma^\prime}\left(\Gamma,T\cup \sum_{k=2}^{v_\Gamma}\sum_{F\in\mathcal{F}_k(\Gamma^{\prime\prime}),|\Gamma_F|=0}F\right)=\sum_{k=2}^{v_\Gamma}\sum_{F\in\mathcal{F}_k(\Gamma)}(\Gamma,F).
\ee
\end{lem}
\begin{remark}
On the right, we have a sum over all $(k\geq 2)$-forests, and therefore a sum over all possible Cutkosky cuts. On the left, we have the same using that the set of all sub-graphs $\Gamma^\prime$  which have loops left intact appear on the lhs of the core Hopf algebra co-product, with intact spanning trees $T$, whilst $\Gamma^{\prime\prime}$ has no loops left intact, $|\Gamma_F|=0$.
\end{remark}
\begin{remark}
The lemma ensures that uncut subgraphs which have  loops can have their loops integrated out. The resulting integrals are part of the integrand of the full graph and its variations determined by the cut edges. Understanding the variations for cuts which leave no loop intecat suffices to understand the variations in the general case.
\end{remark}
We set 
\[
\sum_{k=1}^{v_\Gamma}\sum_{F\in\mathcal{F}_k(\Gamma)}(\Gamma,F)=:\mathrm{Disc}(\Gamma),\]
for the sum of all cuts at a graph.
\subsection{Graph Polynomials and Feynman rules}
We turn to Feynman rules, therefore from graphs and their combinatorial properties to the analytic structure of the amplitudes associated to graphs.
\subsubsection{renormalized Feynman rules}
For graphs of a renormalizable field theory, we get renormalized Feynman rules for an overall logarithmically divergent graph $\Gamma$ ($w(\Gamma)=0$) 
with logarithmically divergent  subgraphs as
\be 
\Phi_R=\int_{\mathbb{P}_\Gamma}\sum_{F\in\mathcal{F}_\Gamma}(-1)^{|F|}\frac{\ln\frac{\Phi_{\Gamma/F}\psi_F+\Phi^0_F\psi_{\Gamma/F}}{\Phi^0_{\Gamma/F}\psi_F+\Phi^0_F\psi_{\Gamma/F}}}{\psi^2_{\Gamma/F}\psi^2_F} \Omega_\Gamma.
\ee
Formula for other degrees of divergence for sub- and cographs can be found in \cite{BrownKr}. In particular, also overall convergent graphs are covered. 
It is important that we use a kinematic renormalization scheme such that tadpole integrals vanish \cite{BrownKr,Borinsky}.

The Hopf algebra in use in the above is based on the renormalization coproduct $\Delta$.

The antipode $S(\Gamma)$ in this Hopf algebra can be written as a forest sum:
\be
S(\Gamma)=-\Gamma-\sum_{F\in\mathcal{F}_\Gamma}(-1)^{|F|} F\times(\Gamma/F).
\ee
\subsubsection{renormalized Feynman rulesfor pairs $(\Gamma,F)$}
We now give the Feynman rules for a graph with some of its internal edges cut. This can be regarded as giving Feynman rules for a pair $(\Gamma,F)$.
\be
\Upsilon_\Gamma^F:=\int\left(\Phi_R(\Gamma^\prime)\prod_{e\in (\Gamma^{\prime\prime}-E_{\Gamma_F})}\frac{1}{P(e)}\prod_{e\in E_{\Gamma_F}}\delta^+(P(e))\right)d^{4|\Gamma/\Gamma^\prime|}k.\label{Upsilon}
\ee
We use Sweedler's notation for the copoduct provided by $\Delta_{\mathrm{core}}$. 

Note that in this formula $\Phi_R(\Gamma^\prime)$ has to stay in the integrand. The internal loops of $\Gamma^\prime$ have been integrated out by $\Phi_R$,
but $\Phi_R(\Gamma^\prime)$ is still an obvious function of loop momenta apparent in $\Gamma/\Gamma^\prime$.
The existence of this factorization into integrated subgraphs times cut cographs is a consequence of Lem.(\ref{cccuts}).
\section{Graph amplitudes and Fubini's theorem}
\label{sec:6}
\vspace{1mm}\noindent
This section just mentions an important point often only implicitly assumed. For a $k$-loop graph $\Gamma$, acting with $\Delta^{k-1}$ gives a sum over $k$-fold tensorproducts of one-loop graphs, each of which corresponding to a possibility to write
$\Phi_R(\Gamma)$ as an iterated integral of one-loop amplitudes.

Below, we study the 3-edge banana as an explicit example.  Each of these possibilities evaluate to the same physical amplitude $\Phi_R(\Gamma)$ uniquely defined on the principal sheet. We need Fubini's theorem for that, and the existence of the operator product expansion  (OPE).

Consider the Dunce's cap. 
$$\duncetrees$$
Its five spanning trees give five choices for a basis for its two loops.
The loop to be integrated out first is a function of the next loop's loop momentum.

If we integrate out first the loop based on three edges, say $l_x:e_b,e_r,e_y$ (a triangle), this is a finite integral which does not need renormalization. The second loop
is $l_y:e_r,e_g$ and carries the overall divergence after integrating $l_x$.

  Still, the counterterm for the subloop based on the two edges $e_r,e_g$ is needed. Indeed, it corresponds to a limit where vertices $b,c$ collapse, a limit in which the Hopf algebra of renormalization needs to provide the expected counterterm, even if the iterated integral $l_y\circ l_x$ has no divergent subintegral.
  
We need the operator product expansion to work precisely in the way it does to have the freedom to use Fubini to come to uniquely defined renormalized amplitudes.
\section{Cutkosky's theorem}
\label{sec:7}
\vspace{1mm}\noindent
In \cite{BlKrCut}, Cutkosky's theorem for a graph $G$  is proven in a particular straightforward way for cuts which leave no loop intact, so  $|G^F|=0$.
We quote
\begin{thm}[Cutkosky]\label{ct} Assume $G_F$ has a non-degenerate physical singularity  at an external momentum point $p'' \in Q_{G_F}$. Let $p \in Q_G$ be an external momentum point for $G$ lying over $p''$. Then the variation of the amplitude $I(G)$ around $p$ is given by Cutkosky's formula
\eq{int}{\text{var}(I(G)) = (-2\pi i)^{\# E_ {G_F}}\int\frac{\prod_{e\in E_G\backslash E_ {G_F}} \delta^+(\ell_e)}{\prod_{e\in E_{G_F}} \ell_e}.
}
\end{thm} 
For the set-up of Cutkosky's theorem in general, we can proceed using Lem.(\ref{cccuts}):\\
-either a  renormalized subgraph is smooth at the threshold divisors of the co-graph: then we can apply Cutkosky's theorem on the nose, and get a variation which is parametrized by the renormalization conditions for these subgraphs with loops;\\
- or it is not smooth. Then, necessarily, its disconitinuity is described by a cut on this subgraph, and hence it has no loops, adding its cut to the cut for the total.
\section{Galois co-actions and symbols}
\label{sec:8}
\vspace{1mm}\noindent
In the above, we allow $|E_{\Gamma_F}|$ internal edges $e\in E_{\Gamma_F}$ to be on-shell. Other authors \cite{Britto} also allow to put internal edges on shell such that they do not separate the graph. This gives no physical variation as one easily proves \cite{BlKrCut}. If one allows for variations of masses as well, in particular in the context of non-kinematical renormalization schemes, then such more general cuts can be meaningful though. Here, we restrict to variations coming from varying external momenta for amplitudes renormalized with kinematic renormalization schemes.

It is of interest to study a coaction.  It is often simply written as a Hopf algebra \cite{Britto,Broedel} and then for general cuts the coproduct $\Delta_{\mathrm{co}}$ has the incidence  form
\[
\Delta_{\mathrm{co}} (\Gamma,E)=\sum_{F\subseteq (E_\Gamma\backslash E)} (\Gamma/(E_\Gamma\backslash (E \cup F)),E)\otimes (\Gamma,E\cup F),
\]
where $E$ is a set of edges, and $F$ a subset of the complement $E_\Gamma\backslash E$.
A pair $(\Gamma,E)$ is to be regarded as a graph with the set of edges $E$ put on-shell.

Restricting on-shell edges to originate from cuts $E_{\Gamma_F}$ allows to read off  
$\Delta_{\mathrm{co}}$ from our lower triangular Hodge matrices in an obvious way, as in the example for the triangle above (Eq.(\ref{coaction})).

This co-product fulfills
\[
\Delta_{\mathrm{co}} \circ \mathrm{Disc}=(\mathrm{Disc}\otimes\mathrm{id})\circ \Delta_{\mathrm{co}},
\]
and
\[
\Delta_{\mathrm{co}} \circ /e=(\mathrm{id}\otimes /e)\circ \Delta_{\mathrm{co}},
\]
where $\mathrm{Disc}$ as before  is the map which sends a Feynman graph to a sum over all cut graphs obtained from all spanning $k$-forests of the graph, $k\geq 2$.
Furthermore, $/e$ is the map which sums over all ways of shrinking an (uncut) edge.

We can make this into a proper coaction (see \cite{Broedel,BrownMotic} for an overview of coactions in the context of Feynman amplitudes)
\[
\rho: V\to V\otimes H, (\rho\otimes \mathrm{id})\circ\rho=(\mathrm{id}\otimes\Delta)\circ\rho.
\]
Here, it suggests itself to take for $V$ the vector space of uncut Feynman graphs, and for $H$ the Hopf algebra of cut graphs with a coproduct as above.

For one-loop graphs, these graphical coactions are in acoordance with the structure of the dilogarithms into which these graphs evaluate as was observed by Abreu et.al., see\cite{Britto}.

The Hodge matrices resulting from the cubical chain complex can be constructed for every graph, and from it the corresponding coaction can be constructed. On the analytic side, the hope is that this confirms the set-up suggested by Brown \cite{BrownMotic}. 

For cuts, these coactions in accordance with the cubical chain complex.

Iterating this coaction in accordance with the co-radical filtration of the Hopf algebra suggests then to define symbols graphically, a theme to be pursued further.
\section{$3$-edge banana}
\label{sec:9}
\vspace{1mm}\noindent
In the process of the blow-up of missing cells in Outer Space a graph polytope is generated \cite{KV3}. Such polytopes combine to jewels. One can  use the blow-up to store the complete sheet structure of the amplitude. 

As an example we consider the construction of colored jewel $J_2$, generated by the 3-edge banana graph. 
\subsection{Existence of the banana monodromy}
Let us first collect an elementary result.
\begin{lem}
The integral for the cut $n$-edge banana 
\[
\int_{\mathbb{M}^{D(n-1)}} d^{D(n-1)}k\prod_{i=1}^n\delta_+(k_i)\equiv \mathrm{Var}(\Phi_R(b_n))
\] exists for any positive integer $D$.
\end{lem}
\begin{proof}
We have $n-1$ loop momenta $k_1,\ldots,k_{n-1}$, and the measure is  $d^{D(n-1)}x=d^Dk_1\ldots d^Dk_{n-1}$. The $\delta_+$-distributions give $n$ constraints.
The $n-1$ integrations over the 0-components of the loop momenta can be constrained by $n-1$ of the $\delta_+$ distributions. The remaining spacelike integrals are over an Euclidean  space $\mathbb{R}^{(D-1)(n-1)}$ and can be done in spherical coordinates. The angle integrations are over a compact sphere, and the one remaining $\delta_+$-distribution fixes the radial integration.
\end{proof}
The argument obviously generalizes to graphs with $n-1$-loops in which $n$ cuts cut all the loops. 

\subsection{$b_3$: three edges}\label{bthree}
We now consider the 3-edge banana $b_3$ on three different masses.
$$\bthree$$
As we will see, the resulting function $\Phi_R(b_3)$ has a structure very similar to the dilogarithm function $\mathrm{Li}_2(z)$.
As a multi-valued function, we can write the latter as 
\[
\mathrm{Li}^{\mathrm{mv}}_2(z)=\mathrm{Li}_2(z)+2\pi\imath\mathbb{Z}\ln{z}+(2\pi\imath)^2\mathbb{Z}\times\mathbb{Z},
\]
or more explicitly
\[
\mathrm{Li}^{\mathrm{mv}}_2(z)(n_1,n_2)=\mathrm{Li}_2(z)+2\pi\imath n_1\ln{z}+(2\pi\imath)^2 n_1n_2.
\]
The variable $n_1$ stores the sheet for the evaluation of the sub-integral $\mathrm{Li}_1(x)$ apparent in the iterated integral representation 
\[
\mathrm{Li}_2(z)=\int^z_0 \frac{\mathrm{Li}_1(x)}{x} dx.
\] 
The sheet for the evaluation of  $\ln{z}$ is stored by $n_2$ and only contributes for $n_1\not=0$. 

Very similarly we will establish the structure of the multi-valued functions assigned to $b_3$ as iterated integrals, with 
$\Phi_R(b_2)^\mathrm{mv}(k^2,m_i^2,m_j^2)$ apparent as a one-loop sub-integral in the two-loop integration assigned to $b_3$ and playing the role of 
$\mathrm{Li}_1(x)$.   

We will find multi-valued functions
\bea
\mathrm{I}^{ij}_k(n_1,n_2)(s) & = & \Phi_R(b_3)(s)+2\pi\imath n_1\int\frac{\mathrm{Var}(\Phi_R(b_2))(k^2;m_i^2,m_j^2)}{(k+q)^2-m_k^2} d^4k\label{iijk}\\
& & +(2\pi\imath)^2\frac{|m_k^2-s||m_i^2-m_j^2|}{2s} n_1n_2.\nonumber
\eea
Here, $i,j,k$ take values in the index set $\{b,y,r\}$ labelling the three different masses, and we regard the three functions  
\[
\mathrm{I}^{by}_r(n_1,n_2)(s)\sim\mathrm{I}^{yr}_b(n_1,n_2)(s)\sim
\mathrm{I}^{rb}_y(n_1,n_2)(s)
\]
as equivalent, with equivalence established by equality along the principal sheet.

Let us come back to $b_3$. Here, the fundamental group has two generators, the interesting question is how to compare this with the generators of monodromy for $\Phi_R(b_3)$ and how this defines corresponding multi-valued functions as above. 

We start by using the fact that we can disassemble $b_3$ in three different ways into a $b_2$ sub-graph, with a remaining edge providing the co-graph.  

Any two of the three edges of the graph $b_3$ can be regarded as a subgraph $b_2\subsetneq b_3$. This is in accordance with the flag structure of $b_3$ generated from an application of $\Delta_{\mathrm{core}}$, which gives a set of  three flags (see \cite{BlKr1}):
\[
\left\{\left( \by,\rtp\right),\left(\yr,\btp\right),\left(\rb,\ytp\right)\right\}.
\]
Let us compute
\[
\mathrm{Var}(\Phi_R(b_3)(s,m_r^2,m_b^2,m_y^2))=\int d^4kd^4l \delta_+(k^2-m_b^2)\delta_+(l^2-m_r^2)
\delta_+((k-l+q)^2-m_y^2),
\]
an integral which exists by the above Lemma.

Using Fubini, this can be written in three different ways in accordance with the flag structure:
\[
\mathrm{Var}(\Phi_R(b_3))=\int d^4k \mathrm{Var}(\Phi_R(b_2))(k^2,m_r^2,m_b^2)
\delta_+((k+q)^2-m_y^2),
\]
or
\[
\mathrm{Var}(\Phi_R(b_3))=\int d^4k \mathrm{Var}(\Phi_R(b_2))(k^2,m_b^2,m_y^2)
\delta_+((k+q)^2-m_r^2),
\]
or
\[
\mathrm{Var}(\Phi_R(b_3))=\int d^4k \mathrm{Var}(\Phi_R(b_2))(k^2,m_y^2,m_r^2)
\delta_+((k+q)^2-m_b^2).
\]
The integrals are well-defined by the above Lemma and give the variation and hence imaginary part of $\Phi_R(b_3)$, which can be obtained from it by a twice subtracted dispersion integral (the renormalized function and its first derivative must vanish as $s=s_0$)
\[
\Phi_R(b_3)(s,s_0)=\frac{(s-s_0)^2}{\pi}\int_0^\infty \frac{\mathrm{Var}(\Phi_R(b_3)(x))}{(x-s)(x-s_0)^2}dx.
\]
Computing $\Phi_R(b_3)$ directly from $\Phi_R(b_2)$ as an iterated integral can be done accordingly in three different ways:
\beas
\Phi_R(b_3)(s,s_0;m_r^2,m_b^2,m_y^2)=\\
=\int d^4k \frac{B_2(k^2,m_r^2,m_b^2)}{(k+q)^2-m_y^2}=\int d^4k \frac{B_2(k^2,m_b^2,m_y^2)}{(k+q)^2-m_r^2}=\int d^4k \frac{B_2(k^2,m_y^2,m_r^2)}{(k+q)^2-m_b^2},
\eeas
with subtractions at $s=s_0$ understood.

There is a subtlety here: this is only correct in a kinematic renormalization scheme where subtractions are done by a Taylor expansion of the integrand around $s=s_0$ \cite{BrownKr,Borinsky}.

This implies that the co-graphs in the flag structure of $b_3$ fulfil 
\[
\Phi_R\left(\btp\right)=\Phi_R\left(\ytp\right)=\Phi_R\left(\rtp\right)=0,
\]
as  tadpoles are independent of the kinematic variable $s$. Hence $b_3$ can be regarded as a primitive element under renormalization.

To study the sheet structure for $b_3$ we now define three different multi-valued functions as promised above
\[
I_k^{ij}=I_k^{ji}=\int\frac{\Phi_R^{\mathrm{mv}}(b_2)(k^2,m_i^2,m_j^2)}{(k+q)^2-m^2_k+i\eta}d^4k,
\]
with subtractions at $s=s_0$ understood as always such that the integrals exist.

For later use in the context of Outer Space we represent them as 
$$\bthreebyr $$

It is convenient to rewrite them as,
\[
I_k^{ij}=\int\frac{\Phi_R(b_2)(k^2,m_i^2,m_j^2)}{(k+q)^2-m^2_k+i\eta}d^4k
+2\pi \imath \mathbb{Z} \sum_{u=1}^3 J^{ij;u}_k,
\]
with 
\[
J^{ij;u}_k=\int d^4 k \frac{J^{ij}_u(k^2)}{(k+q)^2-m_k^2+i\eta},
\]
see Eqs.(\ref{jone},\ref{jtwo},\ref{jthree}).

Note that by the above,
\[
\Phi_R(b_3)=\int\frac{\Phi_R(b_2)(k^2,m_i^2,m_j^2)}{(k+q)^2-m^2_k+i\eta}d^4k,
\]
is well-defined no matter which of the two edges we choose as the sub-graph,
and Cutkosky's theorem defines a unique function $V_{rby}(s)$,
\[
\Im(\Phi_R(b_3)(s))=V_{rby}(s)\Theta(s-(m_r+m_b+m_y)^2).
\]

Before we start computations, we note that we expect that the integrals for $J^{ij;2}_k,J^{ij;3}_k$ have no monodromy as there are no endpoint singularities as the integrand vanishes at the endpoints of the domain of integration, and there are no pinch singularities by inspection. 

But for $J^{ij;1}_k$ we expect monodromy:
The denominator of $V_{ij}$ is $k^2$. So we get monodromy  from the intersection of the zero locus $k^2=0$ (which now lies in the domain of integration as $k^2$ is only bounded from the above by $(m_i-m_j)^2$) and the zero locus $(k+q)^2-m^2_k=0$. 
\subsection{Computation}
We now give computational details for the 3-edge banana graph\footnote{Further computational results can be found in \cite{Vanhove,Weinzierl}.}.
We start by computing 
$\Im(\Phi_R(b_3)(s))=V_{rby}(s)\Theta(s-(m_r+m_b+m_y)^2)$, or equivalently $\Im(J^{ij;3}_k)(s)$.
Consider
\[
\int d^4k \frac{\Theta(k^2-(m_i+m_j)^2)}{2k^2}
\delta_+((k+q)^2)-m_k^2).\]
The $\delta_+$ distribution demands that $k_0+q_0>0$, and therefore we get
\[
\int_{-q_0}^\infty dk_0\int_0^\infty dt\sqrt{t}\frac{\Theta(k_0^2-t-(m_i+m_j)^2)\sqrt{\lambda(k_0^2-t,m_i^2,m_j^2)}}{2(k_0^2-t)} \delta((k_0+q_o)^2-t-m_k^2).
\]
As a function of $k_0$, the argument of the $\delta$-distribution has two zeroes:
$k_0=-q_0\pm\sqrt{t+m_k^2}$.

As $k_0+q_0>0$, it follows $k_0=-q_0+\sqrt{t+m_k^2}$.
Therefore, $k_0^2-t=q_0^2+m_k^2-2q_0\sqrt{t+m_k^2}$.

For our desired integral, we get
\[
\int_0^\infty dt \sqrt{t} \Theta(q_0^2+m_k^2-2q_0\sqrt{t+m_k^2}-(m_i+m_j)^2)\times
\]
\[
\times \frac{\sqrt{\lambda(q_0^2+m_k^2-2q_0\sqrt{t+m_k^2},m_i^2,m_j^2)}}{2(q_0^2+m_k^2-2q_0\sqrt{t+m_k^2})\sqrt{t+m_k^2}}.
\]
 The $\Theta$-distribution requires
 \[
 q_0^2+m_k^2-(m_i+m_j)^2\geq 2q_0\sqrt{t+m_k^2}.
 \]
Solving for $t$, we get
\[
t\leq \frac{\lambda(s,m_k^2,(m_i+m_j)^2)}{4s}
\]
As $t\geq 0$, we must have for the physical threshold $s>(m_k+m_i+m_j)^2$ (which is indeed completely symmetric under permutations of $i,j,k$, in accordance for what we expect for $\Im(\Phi_R(b_3)(s))$).
We then have
\[
\Im(J^{ij;3}_k)(s)=\int_0^{\frac{\lambda(s,m_k^2,(m_i+m_j)^2)}{4s}}\frac{\sqrt{\lambda(s+m_k^2-2\sqrt{s}\sqrt{t+m_k^2},m_i^2,m_j^2)}}{2(s+m_k^2-2\sqrt{s}\sqrt{t+m_k^2})\sqrt{t+m_k^2}}
\sqrt{t}dt.
\]
There is also a pseudo-threshold at $s<(m_k-m_i-m_j)^2$.

Note that the integrand vanishes at the upper boundary $\frac{\lambda(s,m_k^2,(m_i+m_j)^2)}{4s}$, and the integral has a pole at $s=0$ (see below) as for $s=0$ the integral would not converge. The integrand is positive definite in the interior of the integration domain and free of singularities.

The computation of $\Im(J^{ij;2}_k)(s)$ proceeds similarly
and gives
\[
\Im(J^{ij;2}_k)(s)=\int_{\frac{\lambda(s,m_k^2,(m_i+m_j)^2)}{4s}}^{\frac{\lambda(s,m_k^2,(m_i-m_j)^2)}{4s}}\frac{\sqrt{\lambda(s+m_k^2-2\sqrt{s}\sqrt{t+m_k^2},m_i^2,m_j^2)}}{2(s+m_k^2-2\sqrt{s}\sqrt{t+m_k^2})\sqrt{t+m_k^2}}
\sqrt{t}dt.
\]
The integrand vanishes at the upper and lower boundaries. The integrand is positive definite in the interior of the integration domain and free of singularities.

Most interesting is the computation of $\Im(J^{ij;1}_k)(s)$.
It gives
\[
\Im(J^{ij;1}_k)(s)=\int_{\frac{\lambda(s,m_k^2,(m_i-m_j)^2)}{4s}}^\infty \frac{\sqrt{\lambda(s+m_k^2-2\sqrt{s}\sqrt{t+m_k^2},m_i^2,m_j^2)}}{2(s+m_k^2-2\sqrt{s}\sqrt{t+m_k^2})\sqrt{t+m_k^2}}
\sqrt{t}dt.
\]
The integrand vanishes at the lower boundary $\frac{\lambda(s,m_k^2,(m_i-m_j)^2)}{4s}$, and the integral again has a pole at $s=0$. But now the integrand has a pole as $q_0^2+m_k^2-2q_0\sqrt{t+m_k^2}$ is only constrained to $\leq(m_i-m_j)^2$, and hence can vanish in the domain of integration.

This gives us a new variation apparent in the integration of the loop in the co-graph
\[
\mathrm{Var}(J^{ij;1}_k)(s)=\int\sqrt{\lambda(k^2,m_i^2,m_j^2)}\delta(k^2)\delta_+((k+q)^2-m_k^2)d^4k,
\]
which evaluates to
\[
\mathrm{Var}(J^{ij;1}_k)(s)=\overbrace{|m_i^2-m_j^2|}^{\sqrt{\lambda(0,m_i^2,m_j^2)}}\overbrace{\frac{|s-m_k^2|}{2s}\Theta(s-m_k^2)}^{\sqrt{\lambda(s,m_k^2,0)}}.
\]
Adding the contributions, we confirm our expectations Eq.(\ref{iijk}).
\section{Markings and monodromy}
\label{sec:10}
\vspace{1mm}\noindent
Consider the equivalence relation for $b_3$ in Outer Space.
$$\bthreeeq$$
The three possible choices for a spanning tree of $b_3$ result in three different but equivalent markings of $b_3$ regarded as a marked metric graph in (colored) Outer Space.\footnote{For the notion of equivalence in Outer Space refer to \cite{KV1}.} Each different choice corresponds to a different choice of basis for $H^1(b_3)$. The markings given in this picture determine all markings in subsequent picture, where they are omitted. 

The choice of a spanning tree together with an ordering of the roses then determines uniquely a single element in the set of ordered flags of subgraphs, and hence determines one iterated Feynman integral describing the amplitude in question.  

For their evaluation along principal sheets equality of these integrals follows by Fubini,
which gives equality along the principal sheet and  implies  an equivalence relation for evaluation along the non-principal sheets.

On the level of amplitudes, a basis for the fundamental group of the graph, provided by a marking, translates to a basis for the fundamental group for the complement of the threshold divisors of the graph.

Concretely, for the amplitude generated by $b_2$, this is trivial: we have a single generator for the one loop, and this maps to a generator for the monodromy  of the corresponding amplitude.

For $b_3$, we get two generators. A choice as which two edges form the subgraph
$b_2$  then determines the iterated integral. The equivalence of markings 
in Outer Space becomes the Fubini theorem  of iterated integrals for the evaluation along principal sheets, and the corresponding equivalence off principal sheets else.
$$\jewelb$$
Let us have a closer look at this corresponding cell in Outer Space. In the barycenter of the triangle we indicate the graph $b_3$. The green lines form the spine, connecting the graph at the barycenter  to the barycenters of the codimension-one edges of the triangle, which are cells marked by the indicated colored 2-petal roses.

The corners of the triangle are not part of Outer Space, as we are not allowed to shrink loops. In fact, they are blown up to arcs, which are cells populated by graphs for which the choice is obvious as to which two edges are the subgraphs - the corners are the intersections of two edge variables  becoming small as compared to the third. The three different iterated integrals are hence assigned to those arcs in a natural manner. 

For example for the lower left corner, the edge variable $A_b$ is much greater than the edge variables $A_r,A_y$. Along the arc, an equivalence relation operates as well, as the loop formed by edges $e_r,e_y$ can have either of the two edges as its spanning tree. The endpoints of these arcs form the vertices of 
the cell, which is a hexagon. To those vertices we assign roses as indicated, with one small and one big petal, which indicates an order on the petals. 

Note that moving along an arc can be regarded as movement in a fibre given by the chosen subgraph $b_2$, while moving the arc away or toward the barycenter of the triangle is movement  in the base.

Moving from one corner to another utilizes a non-trivial equivalence of our iterated integrals.

We indicate the markings only for some of the graphs, and only for the choice of the red edge as the spanning tree.

Let us have a still closer look at the corners:
$$\jewelpart$$
The equivalence relation is an equivalence relation for the two marked metric graphs, which is indeed coming from an equivalence relation for the two choices of a spanning tree for the 2-edge subgraph on the red and yellow edges, while the corresponding analytic expression is equal for both choices: $I_b^{ry}$.

Moving to a different corner by shrinking the size of the blue edge and increasing say the size of the red edge moves to a different corner while leaving the marking equal. This time we have an equivalence relation between the analytic expressions: 
\[
I_b^{ry}\sim I_r^{by}.
\]
Moving along an arc uses equivalence based on homotopy of the graph, moving along an edge leaves the marking equal, but uses equivalence of analytic expressions $I_{\Gamma/\gamma}^{\gamma}$, here $I_b^{ry}\sim I_r^{by}$ \footnote{In this example the cograph was always a single-edge tadpole whose spanning tree is a single vertex and therefore the equivalence relation from the 1-petal rose $R_1$ to the co-graph is in fact the identity. In general, the decomposition of a graph into a subgraph $\gamma$ and cograph $\Gamma/\gamma$ corresponds to a factorization into equivalence classes for the subgraph and equivalence classes for the cograph familiar from \cite{KV4}.}. 

The complete sheet structure including non-principal sheets is always rather subtle and is reflected by a jewelled space $\mathit{J}_2$ as we discuss now.

A crucial aspect of Outer Space is that cells combine to spaces, and that these spaces provide information, for example on the representation theory of the free group in the case of 
traditional Outer Space, and on the sheet structure of amplitudes in our case.
In particular, the bordification of Outer Space as studied by \cite{KV3}, motivates to glue the cell studied above to a 'jewel':
$$\jewelt $$
The Euclidean simplices are put in a Poincar\'e disk as hyperbolic triangles. We only give markings for a few graphs in the center. To not clutter the figure, we have not given the graphs for the vertices in this figure which are all marked ordered roses as indicated above, by a result of \cite{KV3}.
 
\vspace*{4mm}
\noindent
{\bf Acknowledgment.} It is a pleasure to thank Spencer Bloch for a long-standing collaboration and an uncountable number of discussions. Also, I enjoy to thank David Broadhurst, Karen Vogtmann and  Marko Berghoff for discussions, and the audiences at this 'elliptic conference', and at the Les Houches workshop on 'structures in local quantum field theory' for a stimulating atmosphere.   And thanks to Johannes Bl\"umlein for initiating this KMPB conference at DESY-Zeuthen!

\end{document}